\documentclass[a4paper,10pt]{article}
\usepackage[english]{babel}
\usepackage[margin=32mm]{geometry}
\usepackage{amsmath}
\usepackage{amsfonts}
\usepackage{amssymb}
\usepackage{mathrsfs}
\usepackage{mathtools}
\usepackage{amsthm}
\usepackage{enumerate,enumitem}
\usepackage{bm}
\usepackage{graphicx}
\usepackage{natbib}
\usepackage{comment} 
\usepackage[affil-it]{authblk} 
\usepackage{hyperref}
\usepackage{xcolor}
\hypersetup{
  colorlinks   = true, 
  urlcolor     = blue, 
  linkcolor    = red, 
  citecolor   = blue 
}
\usepackage{color}
\usepackage{caption}
\usepackage{subcaption}

\providecommand{\keywords}[1]
{
 {\small
  \textbf{\textit{Keywords:}}} #1
}

\newtheorem{thm}{Theorem}
\numberwithin{thm}{section}
\newtheorem{lem}[thm]{Lemma}
\newtheorem{prop}[thm]{Proposition}
\newtheorem{cor}[thm]{Corollary}

\theoremstyle{definition}

{\theoremstyle{plain}}
{\theoremstyle{plain}}

\numberwithin{equation}{section}

\newcommand{\dd}{{\rm d}}
\newcommand{\norm}[1]{\|#1\|}
\newcommand{\normal}{{\rm N}}

\newcommand{\half}{\tfrac12}

\newcommand{\E}{{\rm E}}

\newcommand{\eps}{\varepsilon}
\newcommand{\given}{\,|\,}
\newcommand{\real}{\mathbb{R}}

\newcommand{\calF}{\mathcal{F}}

\newcommand{\calA}{\mathcal{A}}
\newcommand{\calB}{\mathcal{B}}

\newcommand{\calG}{\mathcal{G}}

\newcommand{\calY}{\mathcal{Y}}

\newcommand{\sigalg}{{$\sigma$-algebra}}

\newcommand{\PP}{\mathbb{P}}
\newcommand{\tilda}{\widetilde}
\newcommand{\nobs}{{n_{\rm obs}}}

\newcommand{\true}{{\rm true}}
\newcommand{\as}{{\rm a.s.}}
\newcommand{\Leb}{{\rm Leb}}

\title{Posterior uncertainty for kernel density estimates}
\author{\normalfont Dennis Christensen$^{1}$, Torjus Svardal$^{2}$, Leiv R{\o}nneberg$^{2}$, Emil A.~Stoltenberg$^{2,3}$} 
\affil{{\normalfont\normalsize $^1$Norwegian Defence Research Establishment (FFI)}\\
{\normalfont\normalsize $^2$Department of Mathematics, University of Oslo}\\{\normalfont\normalsize $^3$BI Norwegian Business School, Oslo}}	

\date{\today}

\begin{document}
\maketitle

\begin{abstract} Recent work in predictive Bayesian inference has enabled novel Bayesian interpretations of many well-known stochastic one-step-ahead predictive algorithms. In this paper, we study classic kernel density estimation in the predictive Bayesian framework. We prove that their predictive measures converge weakly almost surely---meaning that their associated predictive resampling sequences are almost surely asymptotically exchangeable---and we provide estimators for moments of the limiting random probability measure. We also show that the resampling sequences do not satisfy standard assumptions like being conditionally identically distributed (c.i.d.)~or almost c.i.d.~(a.c.i.d.), thus providing a non-trivial example of a predictive sequence which is not a.c.i.d.~but nevertheless converges weakly almost surely. For Gaussian kernels, we show that the limiting directing measure is almost surely absolutely continuous with respect to the Lebesgue measure, meaning it emits a probability density. This enables us to derive credibility intervals for kernel density estimates, which we illustrate on two real datasets.
\end{abstract}

\begin{small}
\end{small}

\section{Introduction}\label{sec:intro}
Predictive Bayesian inference is a framework for achieving posterior type inference for parameters of interest without passing through the classical prior-likelihood setup of Bayesian modelling. Instead, one concentrates on the one-step-ahead predictive distributions \citep{fong2023martingale,fortini2023prediction,fortini2025exchangeability,berti2025probabilistic}. In this paper, we study kernel density estimation via predictive Bayesian inference. Let $y_{1:\nobs} = (y_1,\ldots,y_{\nobs})$ be our observed data. When these are the realisations of independent and identically distributed (i.i.d.) random variables with density function $f_{\true}$ on the real line, a common nonparametric estimate of $f_{\true}$ is a kernel density estimate of the form
\begin{equation}
f_0(y) = \frac{1}{\nobs h_{\nobs}}\sum_{i=1}^{\nobs} K\bigg(\frac{y - y_i}{h_{\nobs}} \bigg), \quad y \in \real, 
    \label{eq:kde1}
\end{equation}
where $K$ is some density and $h_{\nobs}$ is the bandwidth. Parameters of interest are finite or infinite dimensional mappings of $f_{\true}$ such as $\theta(f_{\true}) = \int y^p f_{\true}(y)\,\dd y$ for some $p \geq 1$, or $\theta(f_{\true}) = f_{\true}(y)$, or indeed $\theta(f_{\true}) = f_{\true}$ itself. Natural estimates of these quantities are obtained by plugging in the kernel density estimate $f_0$ in place of $f_{\true}$, and we write $\theta(f_0) = \theta(y_{1:\nobs})$ to highlight that these are computed based on the observed data. The predictive Bayes point of view is that uncertainty associated with the estimates $\theta(y_{1:\nobs})$ arise because we have not observed the entire (infinite) population, but only the finite sample $y_{1:\nobs}$. Posterior computation then involves simulating the missing observations, which, in our setup, means drawing $Y_1 \sim f_0(y) = f_0(y ; y_{1:\nobs})$, and then recursively sampling\footnote{We write $Y_n$, instead of the more cumbersome of {`}$Y_{\nobs + n}${'}, to ease the notation. This slight double use of notation, i.e., that for $1 \leq i \leq \nobs$, $y_i$ is {\it not} the realisation of $Y_i$, will not cause confusion.} 
\begin{equation}
Y_{n+1}\given Y_1,\ldots,Y_{n}
\sim f_n(y) = f_{n}(y; y_{1:\nobs},Y_1,\ldots,Y_n), \quad \text{for $n \geq 1$}, 
    \label{eq:intrinsic_resamlpling}
\end{equation}
where 
\begin{equation}
f_n(y) = \frac{1}{m h_{m}}\sum_{i=1}^{\nobs} K\big(\frac{y - y_i}{h_m} \big) + 
\frac{1}{m h_{m}}\sum_{i=1}^{n} K\big(\frac{y - Y_i}{h_m} \big), \quad m = \nobs + n.
    \label{eq:kde_intrinsic}
\end{equation}
When the number $n$ of one-step-ahead predictions tends to infinity and $F_n(y) = \int_{-\infty}^y f_n(u)\,\dd u$ converges to some limiting distribution $F_{\infty}$ in an appropriate sense, we (heuristically speaking) obtain a full sample $(y_{1:\nobs}, Y_1,Y_2,\ldots)$, and based on these data compute our parameter of interest, $\theta(F_{\infty}) = \theta(y_{1:\nobs}, Y_1,Y_2,\ldots)$. As the simulated data are random, so is $\theta(F_{\infty})$, and the distribution of $\theta(F_{\infty})$ is called the {\it martingale posterior}, a term introduced by~\citet{fong2023martingale}. The martingale posterior is a posterior in the sense that it is a distribution over the parameter space which depends on the observed data $y_{1:\nobs}$, and, typically, one relies on martingale methods to ensure its existence, hence the concatenation of the two terms. Crucially, when the observed and the simulated data form an exchangeable sequence, the martingale posterior coincides with the traditional Bayesian posterior, providing a theoretical safety net which anchors the method under exchangeability.  

The primary aim of this paper is to show that under certain conditions on the kernel density estimator (KDE), and for various mappings $\theta(\cdot)$, the limiting random variable $\theta(F_{\infty})$ does indeed exist and satisfies desirable properties. A crucial element in proving the existence of a martingale posterior is showing that the predictive distributions $\alpha_n(A) = \Pr(Y_{n+1} \in A \given Y_1,\ldots,Y_n) = \int_{A} f_n(y) \,\dd y,\, A\in \calB(\real)$ converge weakly almost surely ({\as}) to a random probability measure $\alpha$. That is,
\begin{equation}
\Pr (\{ \omega \in \Omega \colon \alpha_n(\cdot,\omega) \Rightarrow \alpha(\cdot,\omega)\} )= 1, 
    \notag
\end{equation}
a notion we denote by $\alpha_n \Rightarrow \alpha$ {\as}, with the bold arrow denoting weak convergence of probability measures.\footnote{And, as usual, we also write $X_n \Rightarrow X$ whenever $X_n \sim P_n$, $X \sim P$, and $P_n \Rightarrow P$.} Weak convergence {\as}~is crucial because it ensures the sequence $(Y_n)_{n \geq 1}$ is {\it asymptotically exchangeable} (see \citet[p.~60]{aldous1985exchangeability}, and also \citet[Theorem~2]{leisen2025weak}). This means that 
\begin{equation}
(Y_{n+1},Y_{n+2},\ldots ) \Rightarrow (Z_1,Z_2,\ldots),
    \notag
\end{equation}
for an exchangeable sequence  $(Z_1,Z_2,\ldots)$ with directing measure $\alpha$. That is, given $\alpha$ the $Z_i$ are conditionally i.i.d.~according to $\alpha$. The existence of the directing measure for the exchangeable sequences is guaranteed by de Finetti{'}s theorem, see for example \citet[Theorem~3.1, p.~19]{aldous1985exchangeability} or \citet[Theorem~1.49, p.~29]{schervish1995theory} for proofs and discussions.

The study of kernel density estimation via predictive Bayesian inference has been treated previously in the literature. \citet[Sect.~4.3]{battiston2025bayesian} studied a recursive version of KDEs with individual bandwidths for each observation in the sequence (instead of letting the bandwidths depend on the sample size, as we do), and required that the bandwidths decayed exponentially fast. Their arguments therefore do not cover the canonical KDEs, which are those of the form~\eqref{eq:kde1} with bandwidth decay rate $n^{-a}$ for some $a>0$. Recently, however, \citet{hilbert2026predictive} studied the same resampling scheme as we do in this paper, with bandwidths $h_n \propto n^{-a}$, and proved asymptotic exchangeability using an argument combining characteristic functions and martingale theory. Furthermore, he showed the rather surprising result that, under weak conditions, the limiting distribution has compact support. Hilbert{'}s paper and the present paper have the same point of departure (i.e., predictive KDEs), but go in somewhat different and complementary directions. The proofs of asymptotic exchangeability are also different.

In the present work, we provide an independent proof (see Section~\ref{sec:asweakconv}) that the predictive distributions $\alpha_n$ associated with KDEs converge weakly {\as}~to a limiting random measures $\alpha$, and hence asymptotic exchangeability of the $Y_n$ sampled from $\alpha_n$. In the predictive Bayes literature, such results are usually proved by verifying that the predictive
sequences are martingales, quasi-martingales \citep{fortini2026principled}, or nonnegative almost supermartingales \citep{battiston2025bayesian}. Predictives that are martingales correspond to $(Y_n)$ sequences that are conditionally identically distributed (c.i.d.), a concept introduced by~\citet{berti2004limit}; while predictives that are form nonnegative almost supermartingales correspond to $Y_n$ sequences that are almost c.i.d.~(a.c.i.d.), which was recently introduced by~\citet{battiston2025bayesian}. As is common in the literature, we may refer to predictives as being c.i.d.~or a.c.i.d.~whenever their corresponding data sequences have any of these properties. We prove explicitly, however, that KDEs with Gaussian kernels $K = \phi$ satisfy neither of these criteria (Prop.~\ref{prop:notacid}). Thus, these KDE predictives provide a non-trivial example of an {\as}~weakly convergent predictive sequence which is neither c.i.d., nor a quasi-martingale, nor a.c.i.d.

In Section~\ref{sec:abscont}, we show that for Gaussian kernels $K = \phi$, the limiting distribution $\alpha$ is absolutely continuous with respect to Lebesgue measure, which means we can derive credibility intervals to assess posterior uncertainty of the kernel density estimate itself. We conjecture that this result holds more generally, and not merely for Gaussian kernels. In Section~\ref{sec:int_to_limit} we consider estimators of various moments of the limiting distribution, and show that these are {\as}~convergent. Finally, in Section~\ref{sec:realdata}, we use our predictive resampling algorithm to analyse two real datasets. All proofs omitted in the main text, along with background and some extra material, can be found in the appendices.

\section{Predictive resampling with kernel densities}\label{sec:predresample_kde}
In this section we study a resampling scheme that is (notationally) less involved than, but closely related to, that described in~\eqref{eq:kde_intrinsic}. The sequence $(Y_n)_{n \geq 1}$ of random variables is defined on a probability space $(\Omega,\calF,\Pr)$, and each $Y_n$ takes its values in $\real$ equipped with the Borel {\sigalg} $\calB$. Let $(\calG_n)_{n \geq 0}$ be the natural filtration of $(Y_n)_{n \geq 1}$, with $\calG_0$ the trivial {\sigalg}. The sequences $(a_n)_{n \geq 0}$ of predictive distributions we study in this section are given by 
\begin{equation}
\alpha_n(A) = \int_A f_n(y) \,\dd y, \quad \text{with}\quad
f_n(y) = \frac{1}{nh_n}\sum_{i=1}^n K \big( \frac{y - Y_i}{h_n}\big) , \quad n \geq 1,
    \label{eq:predkernel1}
\end{equation}
for $A\in\calB$, where $K$ is a density and the bandwidths take the form
\begin{equation}
h_n =cn^{-a}\quad \text{for some $a>0$, $c > 0$}.
    \label{eq:bandwidth}
\end{equation}
The initial predictive $\alpha_0$ is some fixed distribution (that may be formed from the observed data $y_{1},\ldots,y_{\nobs}$). The predictive resampling scheme we study is then
\begin{equation}
Y_{n+1}\given \calG_n \sim \alpha_n, \quad \text{for $n \geq 0$}. 
    \notag
\end{equation}
That the sequence $(Y_n)_{n \geq 1}$ is well-defined, i.e., that there exists a probability measure with $(\alpha_n)_{n \geq 0}$ as its conditionals, is ensured by the Ionescu--Tulcea theorem \citep[Theorem~6.17, p.~116]{kallenberg2002foundations}. The predictives are random probability measures, meaning that $\omega \mapsto \alpha_n(A,\omega)$ is $\calF$-measurable for each $A \in \calB$; and that $A\mapsto \alpha_n(A,\omega)$ is a probability measure on $(\real,\calB)$ for each $\omega$. Note that~\eqref{eq:bandwidth} contains the standard choice of $h_n \propto n^{-1/(2 \beta + 1)}$, for some $\beta \geq 1$ (see \citet[Theorem 1.1, p.~9]{tsybakov2009nonparametric}), which is known to have desirable frequentist properties. A standard choice of kernel $K$, which also forms a key example in the present paper, is the Gaussian kernel, given by $K(z) = \phi(z) = (2\pi)^{-1/2}\exp(-\half z^2)$ for $z \in \real$. 

\subsection{Almost sure weak convergence}\label{sec:asweakconv}
To prove our main weak convergence result for kernel density based predictives, we rely on Theorem~2.2 in~\citet[p.~92]{berti2006almost}. According to that theorem, if $(\alpha_n)_{n \geq 0}$ is a sequence of random probability measures on a Polish space (in this paper we are working with random measures on $\real$ equipped with the Euclidean metric, which is Polish), 
then 
\begin{equation}
\alpha_n \Rightarrow \alpha,\quad \as
    \label{eq:berti1}
\end{equation}
is equivalent to 
\begin{equation}
\text{$\int g(y)\,\alpha_n(\dd y)$ is {\as}~convergent, for all $g \in C_b$},
    \notag
\end{equation}
where $C_b$ is the set of bounded and continuous functions. For the case where $(\alpha_n)_{n \geq 0}$ are random measures on $(\real,\calB)$, we strengthen this (see Appendix~\ref{moreweak}) to an equivalence between~\eqref{eq:berti1} and
\begin{equation}
\text{$\int g(y)\,\alpha_n(\dd y)$ is {\as}~convergent, for all $g \in C_b^{\infty}$},
    \label{eq:ourcond}
\end{equation}
where $C_b^{\infty}$ is the set of all bounded and continuous functions with bounded and continuous derivatives of all orders. Using this equivalence we obtain the following theorem.
\begin{thm}\label{thm:weakconvas} Let $(\alpha_n)_{n \geq 0}$ be as defined in~\eqref{eq:predkernel1}, and suppose that $\int z K(z) \,\dd z = 0$ and $\int z^2K(z) \,\dd z < \infty$. Then there is a random probability measure $\alpha$ such that $\alpha_n \Rightarrow \alpha$ \as. 
\end{thm}
\begin{proof} For $f \in C_b$ write $f = f^{+} - f^{-}$. Then $f^{+}(x) = \half \{|f(x)| + f(x) \}$ and $f^{-}(x) = \half \{f(x) - |f(x)|\}$ showing that both $f^{+}$ and $f^{-}$ are in $C_b$, and they are not nonnegative. It therefore suffices to consider nonnegative $g \in C_b^\infty$. A Taylor expansion yields
\begin{align*}
\E\,\{g(Y_{n+1}) \given \calG_n      \}
 = \frac{1}{n}\sum_{i=1}^n\int g(Y_i)
+ \frac{h_n^2}{2n} \sum_{i=1}^n  \int g^{\prime\prime}(\tilde{y}_i) z^2K(z) \,\dd z \quad {\as}, 
\end{align*}
where $\tilde{y}_i$ is between $h_n z + Y_i$ and $Y_i$. Since $g^{\prime\prime} \in C_b^{\infty}$, and $\int z^2 K(z)\,\dd z < \infty$, 
\begin{equation}
\big| \frac{h_n^2}{2n} \sum_{i=1}^n  \int g^{\prime\prime}(\tilde{y}_i) z^2K(z) \,\dd z \big| \leq C h_n^2, 
    \notag
\end{equation}
where $C$ is a constant only depending on $g^{\prime\prime}$ and $\int z^2 K(z)\,\dd z$. Thus 
\begin{equation}
\E\,\{g(Y_{n+1}) \given \calG_n \} = n^{-1}\sum_{i=1}^n g(Y_i) + O(h_{n}^2)\quad {\as}.
    \notag
\end{equation}
Write $Z_n = n^{-1} \sum_{i=1}^n g(Y_i)$ for $n \geq 1$. Then $Z_n = (1 - 1/n) Z_{n-1} + g(Y_n)/n$, which by the above entails that
\begin{align*}
\E\,\{ Z_n\given\calG_{n-1} \}
\leq Z_{n-1} + Ch_n^2/n\quad {\as}. 
\end{align*}
With $h_n$ chosen according to~\eqref{eq:bandwidth}, we have $\sum_{j=1}^\infty h_j^2/j \leq \sum_{j=1}^\infty j^{-(1+2a)} < \infty$ since $a > 0$, and so the convergence theorem for nonnegative almost supermartingales of \citet{robbins1971convergence} applies (using that $g \geq 0$), implying that $\lim_n Z_n$ exists and is \as~finite. It then follows that 
\begin{equation}
\lim_n \E\, \{ g(Y_{n+1}) \given \calG_n\} = \lim_n\{ Z_n + O(h_n^2) \} = \lim_n Z_n.
    \label{eq:conv1}
\end{equation}
But from Lemma~\ref{lemma:weakconvCbinfty} in the appendix, we have that~\eqref{eq:conv1} is equivalent to there being a random probability measure $\alpha$ such that $\alpha_n \Rightarrow \alpha$ {\as} 
\end{proof}
Note that the proof of Theorem~\ref{thm:weakconvas} does not employ any of the three convergence criteria (c.i.d., quasi-martingale, a.c.i.d.) mentioned in the Introduction. As a brief summary, a sequence of predictives is 
\begin{enumerate}[label=(\roman*)]
    \item c.i.d.~if, for all $A\in\calB$, $\E\,\{\alpha_{n+1}(A)\given \calG_n\} = \alpha_n(A)$ {\as}; 
    \item a quasi-martingale if, for all $A\in\calB$, $\sum_{n }\E\,|\E\,\{\alpha_{n+1}(A) \given \calG_n\} - \alpha_n(A)| < \infty$;
    \item a.c.i.d.~if, for all $A\in\calB$, $\alpha_{n+1}(A) \leq \alpha_n(A) + \xi_n$ {\as}~for a nonnegative and {\as}~summable sequence $(\xi_n)$.
\end{enumerate}
In Appendix~\ref{app:convcriteria} we look at the relation between these three convergence criteria. In particular, we show that predictives satisfying the quasi-martingale criterion are a.c.i.d., but that an a.c.i.d.~sequence is not necessarily a quasi-martingale. It is well known that these three criteria give sufficient conditions for {\as}~weak convergence, but that they are not necessary, which the following lemma confirms.
\begin{prop}\label{prop:notacid} Let $(Y_n)_{n \geq 1}$ be a sequence generated according to~\ref{eq:predkernel1}, with $K = \phi$, where $\phi(z) = \exp(-\half z^2)/\sqrt{2\pi}$ is the Gaussian kernel. Then $(Y_n)_{n \geq 1}$ is not a.c.i.d.
\end{prop}
This theorem is proved in Appendix~\ref{app:notacid}. Note that since $(Y_n)_{n \geq 1}$ is not a.c.i.d., neither is it c.i.d., nor does its predictives form a quasi-martingale. Theorem~\ref{thm:weakconvas} and Proposition~\ref{prop:notacid} thus combine to give a non-trivial example of sequence which is not a.c.i.d., but nevertheless asymptotically exchangeable. It is tempting to conjecture that no non-trivial kernels bring about predictives that are a.c.i.d., but this remains to be verified (or disproved).   

\subsection{Absolute continuity of the limiting distribution}\label{sec:abscont}
When $\alpha_0$ is a absolutely continuous with respect to the Lebesgue measure, i.e., $\alpha_0 \ll \Leb$, then we have, by construction, that $\alpha_n \ll \Leb~\as$~for all $n$. It therefore seems a reasonable supposition that also $\alpha\ll\Leb~\as$, that is, that there exists a nonnegative function $f:\mathbb{R}\times\Omega \to [0, \infty)$ such that $\alpha(\dd y, \omega) = f(y, \omega)\,\dd y$ for almost all $\omega\in\Omega$. This is, however, not immediate. In Appendix~\ref{app:abscont}, we prove the result for Gaussian kernels.

\begin{thm}\label{thm:abscont} Let $(Y_n)_{n \geq 1}$ be a sequence generated according to~\eqref{eq:predkernel1}, with $K = \phi$, where $\phi(z) = \exp(-\half z^2)/\sqrt{2\pi}$ is the Gaussian kernel. Let $\alpha$ be the {\as}~weak limit of the predictives associated with $(Y_n)_{n \geq 1}$ (which exists due to Theorem~\ref{thm:weakconvas}). Then $\alpha \ll {\rm Leb}$ {\as}. 
\end{thm}

We highlight the importance of this theorem. Without it, credibility intervals for the kernel density estimate (when employing Gaussian kernels), like those in Section~\ref{sec:realdata}, have no justification. It is precisely the existence of the limiting density that allows us to assess the posterior uncertainty of the kernel density estimates.

\citet[Sect.~2]{berti2013exchangeable} give sufficient conditions for absolute continuity of the limiting distribution $\alpha$. Namely, that if $\alpha_n \ll {\rm Leb}$ {\as} for each $n$ and $\norm{\alpha_n - \alpha}_{\rm TV} \to 0$ {\as}, then $\alpha \ll {\rm Leb}$ a.s.,\footnote{In fact, Theorem~1 in~\citet[p.~2094]{berti2013exchangeable} says that whenever $(Y_n)_{n \geq 1}$ is c.i.d., the implication also goes in the other direction. Here, we are only interested in the forward direction.} where $\norm{\cdot}_{\text{TV}}$ denotes the total variation norm. Our proof of Theorem~\ref{thm:abscont}, however, relies on a different argument, and so the question as to whether $\norm{\alpha_n - \alpha}_{\rm TV} \to 0$ {\as}~remains an open problem.

\subsection{Uniform integrability and integration to the limit}\label{sec:int_to_limit} 
In studying the predictively resampled sequences $(Y_n)_{n \geq 1}$ it is important to control the moments of $f_n$. In particular, we need to establish that $(Y_n^p)_{n \geq 1}$ is uniformly integrable. The following identity, whose proof can be found in Appendix~\ref{app:momentlemma}, is useful in this regard.
\begin{lem}\label{momentlemma} For an integer $p \geq 1$, suppose that $\E\,|Y_1|^p$ is finite, and let $K$ be a kernel so that $\int |z|^p K(z) \,\dd z$ is finite. Then 
\begin{equation}
\E\,Y_{n+1}^p    
= c_p h_n^p + \sum_{\ell=1}^{n-1} \frac{h_{\ell}^p}{\ell + 1} 
+ \sum_{j=1}^{p-1} \binom{p}{j} \bigg( \sum_{\ell=1}^{n-1}\frac{h_{\ell}^{p-j} a_{\ell,j} }{\ell+1}\bigg) c_{p-j}
+ \sum_{j=1}^{p-1} \binom{p}{j} h_n^{p-j}c_{p-j}a_{n,j}
+ \E\, Y_1^p,
\notag
\end{equation}
where $a_{k,j} = k^{-1}\sum_{i=1}^k\E\, Y_i^j$ and $c_m = \int z^m K(z) \,\dd z$, so $c_m=0$ for odd $m$ if $K$ is symmetric. Moreover, the same identity holds if we replace $\E\, Y_i^j$ and $\int z^jK(z)\,\dd z$ by $\E\,|Y_i|^j$ and $\int |z|^j K(z)\,\dd z$, respectively, throughout.
\end{lem}
One thing that is important to note, and that is important for the proof of the next lemma, is that with bandwidths $h_n$ of the form~\eqref{eq:bandwidth}, 
\begin{equation}
\sum_{n = 1}^{\infty} \frac{h_{n}^j}{n + 1} 
< \infty,\quad \text{for any $j \geq 1$}. 
    \notag
\end{equation}
We have the following lemma.
\begin{lem}\label{lemma:UI} If the conditions of Lemma~\ref{momentlemma} are in force, then $\sup_{n}\E\,|Y_n|^p < \infty$. In particular, $(Y_n^j)_{n \geq 1}$ is uniformly integrable for each $j = 1,\ldots,p-1$.
\end{lem}
\begin{proof} For $p = 1$, the absolute-value version of the expression in Lemma~\ref{momentlemma} reads $\E\,|Y_{n+1}| = c_1 h_n + \sum_{\ell + 1}^{n-1}h_{\ell}/(\ell + 1) + \E\,|Y_1|$. Since $h_n \leq h_1$ and the sum $\sum_{\ell + 1}^{n-1}h_{\ell}/(\ell + 1)$ is convergent as $n \to \infty$, we conclude that $\sup_{n}\E\,|Y_{n+1}| \leq c_1 h_1 + \sum_{\ell = 1}^{\infty} h_{\ell}/(\ell + 1)+ \E\,|Y_1| < \infty$. This yields the base case of an induction proof. Suppose that $\sup_{n}\E\,|Y_n|^j < \infty$ for $j = 1,\ldots,p-1$. We show that then $\sup_{n}\E\,|Y_n|^p < \infty$ as well. Let $\breve{\alpha}_{m,j} = m^{-1}\sum_{i=1}^m \E\,|Y_i|^j$. Then, using the induction hypothesis,    
\begin{align*}
\sup_{n}   \sum_{\ell=1}^{n-1}\frac{h_{\ell}^{p-j} \breve{a}_{\ell,j} }{\ell+1}  
\leq \sup_{n}\big(\sup_{\ell \leq n-1}\breve{a}_{\ell,j}\big) \sum_{\ell = 1}^{n-1}\frac{h_{\ell}^{p-j}}{\ell + 1}
\leq \sup_n \breve{a}_{n,j} \sum_{\ell = 1}^{\infty}\frac{h_{\ell}^{p-j}}{\ell + 1} < \infty, 
\end{align*}
because $\sup_n \breve{a}_{n,j} = \sup_n n^{-1}\sum_{i=1}^n \E\, |Y_i|^j \leq \sup_n \E\, |Y_n|^j < \infty$ by the hypothesis. The remaining terms in the expression for $\sup_{n}\E\,|Y_{n+1}|^p$ are finite by similar arguments. 
\end{proof}

For random variables $X_n$ converging in distribution to a random variables $X$, a standard result says that if $(X_n)_{n\geq 1}$ is uniformly integrable, then $X$ is integrable and $\E\, X_n \to \E\, X$ \citep[Theorem~25.12, p.~338]{billingsley96pm}. For our purposes, we need an analogous result for {\as}~weak convergence to show that the moments are \as~convergent. If the sequence $(Y_n)_{n \geq 1}$ is exchangeable with directing measure $P$, then its predictives are of the form $\alpha_n(A) = \E\,\{P(A) \given \calG_n\}$. Given that $\E\, |Y_n|^p < \infty$, we have that $\E\,\{ Y_{n+1}^p \given \calG_n\} = \E\, \{ \int y^p \,\dd P(y)  \given \calG_n  \}$, and so 
\begin{equation}
\E\,\{ Y_{n+1}^p \given \calG_n\} \to \int y^p \,\dd P(y) ,\; \as.
    \notag
\end{equation}
This is essentially paraphrasing Doob{'}s consistency theorem, see, e.g.,~\citet[Theorem~2.2]{miller2018detailed} or \citet[ch.~6.2]{ghosal2017fundamentals}. If the sequence $(Y_n)_{n \geq 1}$ is not exchangeable, but has {\as}~weakly convergent predictives, we need a conditional notion of uniform integrability. If $(Z_n)_{n \geq 1}$ is a sequence of random variables with predictives $(\eta_n)_{n \geq 0}$, we say that $(Z_{n+1},\eta_n)_{n \geq 0}$ is {\as}~uniformly integrable if 
\begin{equation}
\lim_{M \to \infty}\sup_{n}\int_{|z|\geq M} |z| \,\eta_n(\dd z) = 0\quad {\as}.
 \label{eq:UIas}
\end{equation}

\begin{lem}\label{lemma:momentconvergence} If $\alpha_n \Rightarrow \alpha$ {\as} and $(Y_n^p,\alpha_n)_{n \geq 1}$ is {\as}~uniformly integrable, 
then 
\begin{equation}
\int y^p \,\dd \alpha_n \to \int y^p \,\dd \alpha,\quad {\as}.
    \notag
\end{equation}
\end{lem}
\begin{proof} This is direct from, for example, Theorem~25.12 in~\citet[p.~338]{billingsley96pm}, using that the intersection of two null sets is null.  
\end{proof}

The following lemma gives a sufficient condition for $(Y_{n+1},\alpha_n)_{n \geq 0}$ being {\as}~uniformly integrable. 
\begin{lem}\label{lemma:suffasUI} Let $(Z_n)_{n \geq1}$ be a sequence of random variables with predictives $(\eta_n)_{n \geq 0}$. If $\sup_n \E\, |Z_n|^{1 + \delta} < \infty$ for some $\delta >0$, then $(Z_n,\eta_n)_{n \geq 1}$ is {\as}~uniformly integrable.
\end{lem}
\begin{proof} Let $A_n = \{  \int |z|^{1 + \delta}\,\eta_n(\dd z) = \infty\}$. Then
\begin{align*}
\Pr(A_n) & = \Pr( \bigcap_{M > 0} \{\omega \in \Omega \colon \int |z|^{1 + \delta}\,\eta_n(\dd z,\omega) \geq M\} )
\leq \Pr( \int |z|^{1 + \delta}\,\eta_n(\dd z) \geq M )\\
& \leq \frac{1}{M}\E\, |Z_{n+1}|^{1 + \delta}
\leq \frac{1}{M}\sup_n\E\, |Z_{n+1}|^{1 + \delta},
\end{align*}
so we can choose $M = M_n$ so that $\Pr(A_n) \leq 2^{-n}$ for example, or whatever ensuring summability of $\sum_{n}\Pr(A_n) < \infty$. The Borel--Cantelli lemma then gives 
\begin{equation}
\Pr\big( \limsup_n \int |z|^{1 + \delta}\,\eta_n(\dd z) = \infty \big) = \Pr( \limsup_{n} A_n) = 0,  
    \notag
\end{equation}
and so $\limsup_n \int |z|^{1 + \delta}\,\eta_n(\dd z) < \infty$ {\as}, which implies~\eqref{eq:UIas}. 
\end{proof}

Define the empirical moments 
\begin{equation}
\widehat{\mu}_n^{(p)} = \frac{1}{n}\sum_{i=1}^n Y_i^p,\quad \text{for $p \geq 1$}.
    \notag
\end{equation}
The next lemma establishes that $\widehat{\mu}_n^{(p)}$ is almost surely convergent, provided $\int |z|^pK(Z) \,\dd z$ is finite. The proof relies on the quasi-martingale convergence theorem (see for example Theorem~26.20 in~\citet[p.~532]{kallenberg2002foundations} coupled with the supermartingale convergence theorem, e.g., Corollary~27.1 in~\citet[p.~230]{jacod2004probability}). Recall that a sequence $(M_n)_{n}$ of adapted integrable random variables is a quasi-martingale with respect to the filtration $(\calF_n)_{n}$ if $\sum_n \E\, \big|\E\,\{M_{n+1} \given \calF_n\} - M_n \big| < \infty.$ The quasi-martingale convergence theorem says that if $(M_n)_{n \geq 1}$ is a quasi-martingale and $\sup_n \E\,|M_n| < \infty$, then $M_n \to M$ {\as}, with $\E\,|M| < \infty$. If $(M_n)_n$ happens to be uniformly integrable then the convergence also takes place in $L^1$, i.e., $\E\,|M_n - M| \to 0$. 

\begin{lem} Suppose that $(Y_n^p)_{n\geq 1}$ is uniformly integrable. Then, for $j = 1,\ldots,p$, 
\begin{equation}
\widehat{\mu}_n^{(j)} \to\mu^{(j)}\quad {\as},
    \notag
\end{equation}
for integrable random variables $\mu^{(1)},\ldots,\mu^{(p)}$.
\end{lem}
\begin{proof} Write $\E\,\{ \widehat{\mu}_{n+1}^{(p)} \given \calG_n\} = \{1 - 1/(n+1)\}\widehat{\mu}_{n}^{(p)} + \{1/(n+1)\}\E\,\{Y_{n+1}^p \given \calG_n\}$. The expectation on the right is $\E\,\{Y_{n+1}^p \given \calG_n\}  = 
\widehat{\mu}_{n}^{(p)} + 
\sum_{j=0}^{p-1}\binom{p}{j} h_{n}^{p-j} c_{p-j}\widehat{\mu}_{n}^{(j)}$, which inserted in the previous expression gives
\begin{equation}
\E\,\{ \widehat{\mu}_{n+1}^{(p)} \given \calG_n\}
= \widehat{\mu}_n^{(p)} + 
\sum_{j=0}^{p-1}\binom{p}{j} \frac{h_{n}^{p-j}}{n+1} c_{p-j}\widehat{\mu}_{n}^{(j)}. 
    \label{eq:expect_momentestimator}
\end{equation}
In particular, writing $\widehat{\mu}_{n+1} = \widehat{\mu}_{n+1}^{(1)}$, 
\begin{equation}
\E\,\{ \widehat{\mu}_{n+1} \given \calG_n\}
= \widehat{\mu}_n + \frac{h_n}{n+1}c_1,
    \notag
\end{equation}
where, if $K$ is symmetric, $c_1 = \int z K(z) \,\dd z = 0$, but we do not assume that here. Thus
\begin{equation}
\sum_n \E\, \big|\E\,\{ \widehat{\mu}_{n+1} \given \calG_n\} - \widehat{\mu}_n \big|
\leq \sum_n \frac{h_n}{n+1}c_1 < \infty, 
    \notag
\end{equation}
This shows that $(\mu_n)$ is a quasi-martingale, and $\sup_n \E\,|\mu_n| < \infty$ since $\sup_n \E\,|Y_n| < \infty$. The quasi-martingale convergence theorem then yields $\widehat{\mu}_n \to \mu = \mu^{(1)}$ {\as} This will serve as our base case. Next, make the (strong) induction hypothesis that $\widehat{\mu}_n^{(j)} \to \mu^{(j)}$ {\as}, for $j = 1,\ldots,p-1$. From~\eqref{eq:expect_momentestimator} we get that
\begin{equation}
\sum_n \E\, \big| \E\,\{ \widehat{\mu}_{n+1}^{(p)} \given \calG_n\}
- \widehat{\mu}_n^{(p)}      \big|
\leq 
\sum_{j=0}^{p-1}\binom{p}{j} \sum_n \frac{h_{n}^{p-j}}{n+1} c_{p-j}\E\, |\widehat{\mu}_{n}^{(j)}|,
    \notag
\end{equation}
where the right hand side is convergent. This is because $\widehat{\mu}_{n}^{(j)} \to \mu^{(j)}$ {\as}, and since $(\widehat{\mu}_{n}^{(j)})_n$ is uniformly integrable $\E\,|\widehat{\mu}_{n}^{(j)} - \mu^{(j)}| \to 0$ as $n$ tends to infinity. Given $\eps > 0$, let $m < n$ be so that $\E\,|\widehat{\mu}_{\ell}^{(j)} - \mu^{(j)}| < \eps$ for all $\ell \geq m$. Then   
\begin{align*}
0 & \leq \sum_{\ell=1}^n \frac{h_{\ell}^{p-j}}{\ell+1} \E\, |\widehat{\mu}_{\ell}^{(j)}|    
- \sum_{\ell=1}^m \frac{h_{\ell}^{p-j}}{\ell+1} \E\, |\widehat{\mu}_{\ell}^{(j)}|
= \sum_{\ell=m+1}^n \frac{h_{\ell}^{p-j}}{\ell+1} \E\, |\widehat{\mu}_{\ell}^{(j)}|\\
& \leq  \sum_{\ell=m+1}^n \frac{h_{\ell}^{p-j}}{\ell+1} \E\, |\widehat{\mu}_{\ell}^{(j)} - \mu^{(j)}| + \sum_{\ell=m+1}^n \frac{h_{\ell}^{p-j}}{\ell+1} \E\, |\mu^{(j)}|
 \leq   \sum_{\ell=m+1}^n \frac{h_{\ell}^{p-j}}{\ell+1}\,\eps  + \sum_{\ell=m+1}^n \frac{h_{\ell}^{p-j}}{\ell+1} \E\, |\mu^{(j)}| ,
\end{align*}
which tends to zero as $n,m \to \infty$ since $\sum_{\ell =1}^\infty h_{\ell}^{p-j}/(\ell + 1)$ is convergent. This shows that $\sum_{\ell = 1}^n h_{\ell}^{p-j}\E\, |\widehat{\mu}_{\ell}^{(j)}|/(\ell + 1)$ is Cauchy, hence convergent. The quasi-martingale convergence theorem then yields $\widehat{\mu}_n^{(p)} \to \mu^{(p)}$ {\as}
\end{proof}
In the two above lemmata we have shown if $\alpha_n \Rightarrow \alpha$ {\as}~and $(Y_n^p,\alpha_n)_{n \geq 1}$ is {\as} uniformly integrable (as defined in~\eqref{eq:UIas}), then $\E\,\{Y_{n+1}^p \given \calG_n\}\to \int y^p \,\dd \alpha$ {\as}; and that if $(Y_n^p)_{n\geq 1}$ is uniformly integrable, then $\widehat{\mu}_n^{(p)} \to \mu^{(p)}$ {\as}~for some integrable random variable $\mu^{(p)}$. Ideally, we would like $\mu^{(p)}$ to be a version of $\int y^p \,\dd \alpha$. We achieve this under somewhat stronger conditions. The following lemma is an immediate consequence of a Lemma~2 in \citet[p.~534]{berti2011central}. 
\begin{lem}\label{lemma:meanconv} If the conditions of Lemma~\ref{lemma:momentconvergence} are in force, and $\sum_{n=1}^{\infty}n^{-2}\E\,Y_n^{2p} < \infty$, then $\widehat{\mu}_n^{(p)} \to \int y^p \,\dd \alpha$ {\as}, in particular $\mu^{(p)} = \int y^p \,\alpha(\dd y)$ {\as}. 
\end{lem}
\begin{proof} The proof is the same as in \citet{berti2011central}, save for notational differences. We include it for intelligibility. Write 
\begin{equation}
\widehat{\mu}_n^{(p)} = \frac{1}{n}\sum_{k=1}^n Y_k^p
= \frac{1}{n}\sum_{k=1}^n \int y^p\, \alpha_{k-1}(\dd y) + \frac{1}{n} \sum_{k=1}^n k \frac{Y_k^p - \int y^p\,\alpha_{k-1}(\dd y) }{k}.
    \notag
\end{equation}
The first sum on the right tends {\as}~to $\int y^p \,\alpha(\dd y)$ as a consequence of Lemma~\ref{lemma:momentconvergence} and the Toeplitz lemma, whereas the second sum tends {\as}~to zero by Kronecker{'}s lemma provided $M_n \coloneqq\sum_{k=1}^n  k^{-1}\{Y_k - \int y^p\,\alpha_{k-1}(\dd y)\}$ is {\as}~convergent as $n \to \infty$. The summands $k^{-1}\{Y_k - \int y^p\,\alpha_{k-1}(\dd y)\}$ are martingale increments, so a sufficient condition for Kronecker{'}s lemma to be applicable is that the martingale $(M_n)_n$ is uniformly integrable, for which a sufficient condition is that $\sup_{n}\E\,M_n^2 \leq \sum_{n=1}^{\infty}n^{-2}\E\,Y_n^{2p}$ is finite.   
\end{proof}

\section{A kernel density estimate and its intrinsic predictives}\label{sec:intrinsicpreds} In general, with $Y_1,Y_1,\ldots$ a sequence of random variables taking values in a topological space $(\calY,\calB)$, the construction of a predictive distribution for $Y_{n+1}$ given $Y_1,\ldots,Y_n$ is done according to some algorithm $\calA$, that is
\begin{equation}
\calA \colon \bigcup_{m \geq 1} \calY^m \to \{\text{all probability measures on $(\calY,\calB)$}   \}. 
    \notag
\end{equation}
In the predictive Bayes setting we start out with some observed data $y_{1},\ldots,y_{\nobs}$, and construct a predictive distribution $\alpha_0 = \calA(y_1,\ldots,y_{\nobs})$ based on these, where $\calA$ is some algorithm, for example, as in this paper, a kernel density estimate. To sample from the martingale posterior, we do predictive resampling $Y_{n+1}\given \calG_n \sim \alpha_n,\,n \geq 1$ for a sequece $\alpha_1,\alpha_2$ of predicitive distributions. If this sequence of predicitive distributions are constructed using the same algorithm $\calA$ used to construct the initial $\alpha_0$, that is
\begin{equation}
\alpha_n = \calA(y_1,\ldots,y_{\nobs},Y_1,\ldots,Y_n), \quad n \geq 1, 
    \notag
\end{equation}
then we say that $\alpha_1,\alpha_2$ are the {\it intrinsic predictives} of $\alpha_0$, or simply, that $(\alpha_n)_{n \geq 0}$ is a sequence of intrinsic predictives. 

The sequence of predictive distributions worked with in Section~\ref{sec:predresample_kde} are not intrinsic. A kernel density estimate $f_0$ of the form given in~\eqref{eq:kde1} has intrinsic predictives $\alpha_n(A) = \int_A \tilda{f}_n(y)\,\dd y$ for $n \geq 1$ where the densities $\tilda{f}_n$ are (where we for convenience changing the notation a bit compared to~\eqref{eq:kde_intrinsic})
\begin{equation}
\tilda{f}_n(y) 
= \frac{1}{(\nobs + n)h_n} \sum_{i=1}^{\nobs}K \big(\frac{y - y_i}{h_n} \big)
+ \frac{1}{(\nobs + n)h_n} \sum_{i=1}^{n}K \big(\frac{y - Y_i}{h_n} \big), \quad n \geq 0,
    \label{eq:naturalpreds}
\end{equation}
with $h_n \propto (\nobs + n)^{-a}$ for some $
a > 0$, and $Y_{n+1}\given \calG_n \sim \alpha_n$. This differs from the previous section where it was {\it not} required that the (nonrandom) initial distribution $\alpha_0$ had anything to do with the (random) kernel densities $f_1,f_2,\ldots$, with these as defined in~\eqref{eq:predkernel1}. The main results of the previous section do, however, go through almost unchanged. To see this, introduce $f_0(y,h_n) = \nobs^{-1}\sum_{i=1}^{\nobs} h_n^{-1}K((y - y_i)/h_n)$, so $f_0(y,h_0) = f_0(y)$, and write
\begin{equation}
\tilda{f}_n(y) = w_n f_0(y,h_n)
+ (1 - w_n) f_n(y), \quad w_n = \frac{\nobs}{\nobs + n}, \quad n \geq 0, 
    \notag
\end{equation}
where $f_n$ is as defined in~\eqref{eq:predkernel1}. Then
\begin{equation}
\E\, Y_{n+1}^p = w_n \int y^p f_0(y,h_n) \,\dd y
+ (1 - w_n) \int y^p f_n(y)\,\dd y.
\notag
\end{equation}

Using the inequality $|a + b|^p \leq 2^p |a|^p + 2^p |b|^p$, we see that
\begin{align*}
\E\, |Y_1|^p & = \int |y|^p f_0(y)\,\dd y
 = \frac{1}{\nobs}\sum_{i=1}^{\nobs} \int |h_0 z + y_i|^pK(z)\,\dd z\\
 & \leq 2^p h_0^p \int |z|^p K(z) \,\dd z + \frac{2^p}{\nobs}\sum_{i=1}^{\nobs}|y_i|^p, 
    \notag
\end{align*}
and therefore $\E\, |Y_1|^p$ is finite whenever $\int |z|^pK(z)\,\dd z$ is. Similarly, since $h_n$ is bounded above, $\sup_n\int |y|^p f_0(y, h_n)\,\dd y$ is also finite whenever $\int |z|^pK(z)\,\dd z$ is. This gives the following corollary to Lemma~\ref{lemma:UI}.

\begin{cor}\label{cor1:naturalpreds} Let $(Y_n)_{n \geq 1}$ be a sequence such that $Y_{n+1}\given \calG_n \sim \tilda{f}_n$ for $n \geq 0$, and $K$ be a kernel such that $\int |z|^p K(z) \,\dd z$ is finite. Then $(Y_n^j)_{n \geq 1}$ is uniformly integrable for $j = 1,\ldots,p-1$. 
\end{cor}    

The weak convergence proof of Theorem~\ref{thm:weakconvas} also goes through almost unchanged. The only two differences being that the constant $C$ appearing in that proof now also depends on $\nobs$; and that for $g \in C_b^\infty$ with $g$ nonnegative, $\E\, \{g(Y_{n+1}) \given \calG_n\} = w_n \sum_{i=1}^n g(Y_i) + O(h_n^2)$, with $w_n$ instead of $1/n$. But since $w_n \sum_{i=1}^n g(Y_i) = (1 - w_n)\sum_{i=1}^{n-1} g(Y_i) + w_n g(Y_n)$, the almost supermartingale property follows as in the proof of Theorem~\ref{thm:weakconvas}. Similarly, the proof of absolute continuity of the limiting distribution (Theorem~\ref{thm:abscont}) also goes through, \emph{mutatis mutandis}.

\begin{cor}\label{cor:weakconv_naturalpred} Let $\alpha_n(A) = \int_A \tilda{f}_n(y)\,\dd y$ for $A \in \calB$ and $n \geq 0$, where $\tilda{f}_n$ is as in~\eqref{eq:naturalpreds}, with a kernel $K$ such that $\int z K(z)\,\dd = 0$ and $\int z^2 K(z)\,\dd < \infty$. Then there is a random probability measure $\alpha$ such that $\alpha_n \Rightarrow \alpha$ {\as}.   
\end{cor}
The integration to the limit results of Section~\ref{sec:int_to_limit} can similarly be extended to the intrinsic predictives setting. 

\section{Real data examples}\label{sec:realdata}
We apply the predictive resampling setup to the 'Old Faithful' dataset \citep{Azzalini1990}, consisting of ${\nobs}=272$ waiting times between eruptions of the Old Faithful geyser in Yellowstone National Park, Wyoming, United States. Using the bandwidth sequence $h_n=n^{-1/5}$ and Gaussian kernels, we start from our initial dataset $y_1,\ldots,y_{\nobs}$ recursively sample $Y_1, Y_2, \dots, Y_N$ according to~\eqref{eq:intrinsic_resamlpling}. We do so $N= 100,000$ steps forward, in parallel, for $B=1000$ simulation paths. At the end, we compute the KDE for each of the $B$ sample paths. The results are given in Figure~\ref{fig:old_faith}, displaying an estimate of the posterior mean, along with pointwise pointwise credible bands, alongside the KDE computed from the original $\nobs=272$ datapoints.

We note that the posterior mean resembles the original KDE, but perhaps appears a bit smoother. This is further confirmed by adding to the figure an alternative KDE on the original data with a bandwidth $h_n=cn^{-1/5}$, and empirically tuning the constant $c$ to match the posterior mean. We find that a value of around $c=1.5$ appear optimal, and obtain a good match with the martingale posterior mean as indicated by the green dashed line. In Appendix \ref{app:posteriormean} we further make an elementary attempt to elucidate the behaviour of the posterior mean, showing that it can be studied as an infinite product of linear operators acting on the empirical measure $\PP_N$. As for the credible intervals, we note that these do a good job overall in covering the original KDE, and that the width of the intervals are reasonably within the data range. This is in contrast to a Dirichlet Process Mixture Model (DPMM) fit of the data as shown in Figure \ref{fig:old_faith_dpmm}. Here the uncertainty is at its lowest in the region where the data is most scarce (indicated by the shading on the tick marks in the x-axis). The DPMM was fit using the default implementation in the \texttt{dirichletprocess} package in R \citep{CRAN::dirichletprocess}.
\begin{figure}
     \centering
          \begin{subfigure}[b]{0.85\textwidth}
         \centering
         \includegraphics[width=\textwidth]{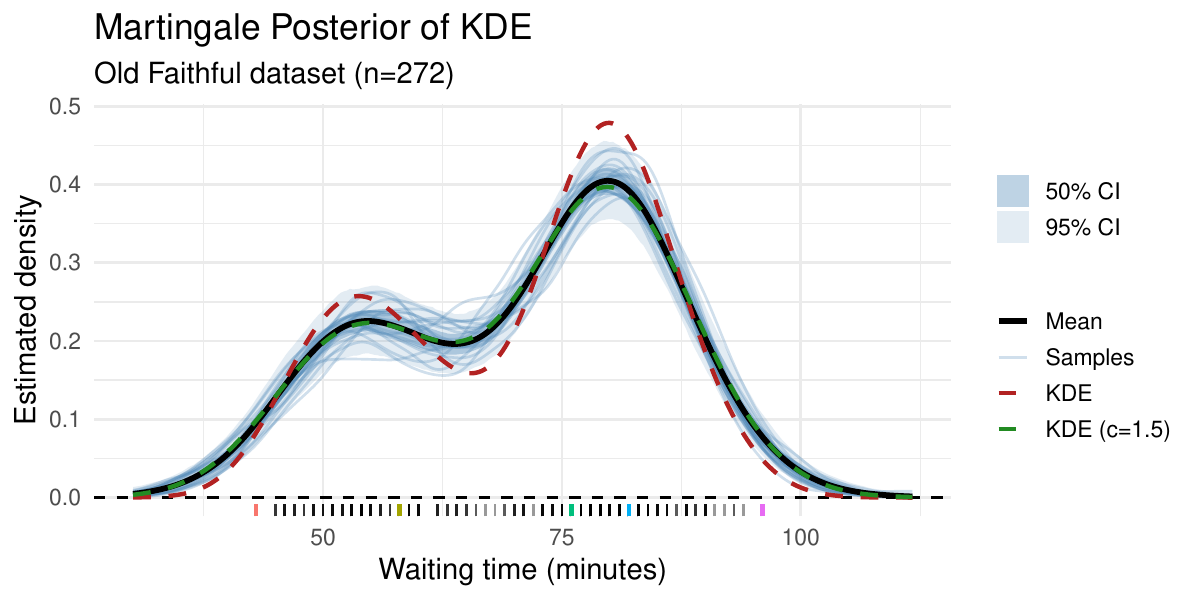}
         \caption{}
         \label{fig:old_faith}
              \end{subfigure}
         \\
     \begin{subfigure}[b]{0.85\textwidth}
         \centering
         \includegraphics[width=\textwidth]{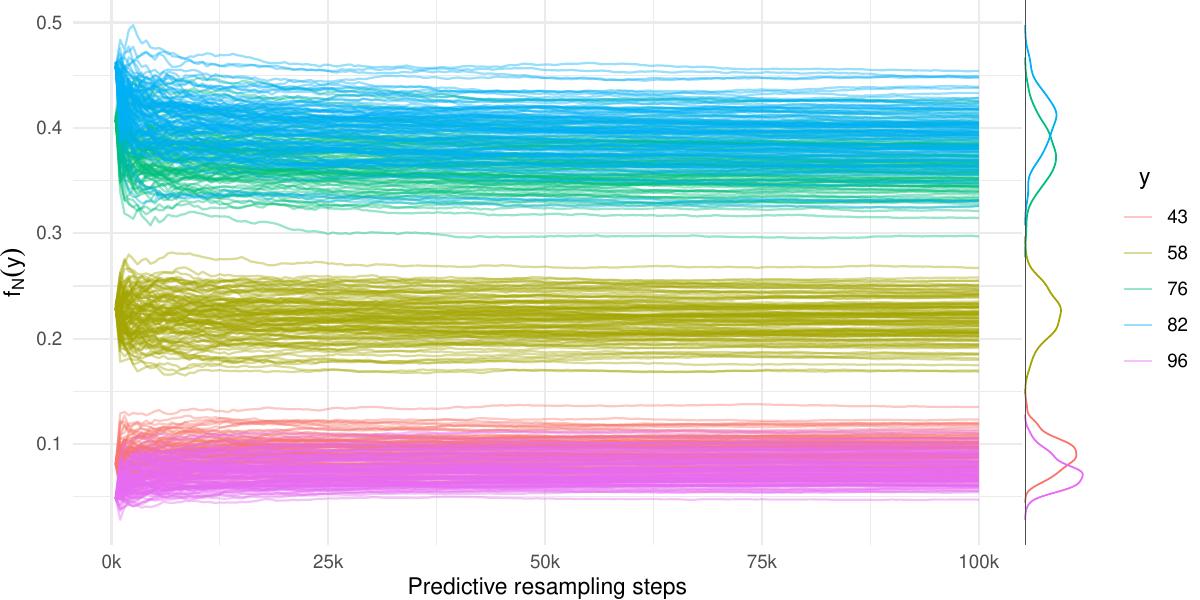}
         \caption{}
         \label{fig:old_faith_conv}
     \end{subfigure}
        \caption{(a) Martingale posterior of KDE applied to the 'Old Faithful' dataset of $\nobs = 272$ waiting times. The martingale posterior mean (black line), 50\% and 95\% credible regions (shaded blue) are displayed, along with 20 samples from the posterior. Additionally, two KDE computed on the original data is shown (red, dashed), using the bandwidth $h_n=n^{-1/5}$, along with an extra smoothed version $h_n=1.5n^{-1/5}$. The tick marks on the x-axis indicate locations of the original data, and the colouring corresponds to values of $y_i$ at which the density has been evaluated in the convergence plot in the panel below. (b) Convergence plot for the martingale posterior of KDEs in the 'Old Faithful' dataset. Each colour corresponds to a specific $y$ location, and we plot the densities $f_{N}(y)$ as a function of $N$. In the right of the figure we compute KDEs over the $B$ sample paths at the final resampling step, showing the uncertainty in the density's value. The five different $y$-values are chosen as the minium, 25\%, 50\%, 75\% and maximum value of the original data.}
        \label{fig:old_faith_MP_joint}
\end{figure}

\begin{figure}
    \centering
    \includegraphics[width=0.85\textwidth]{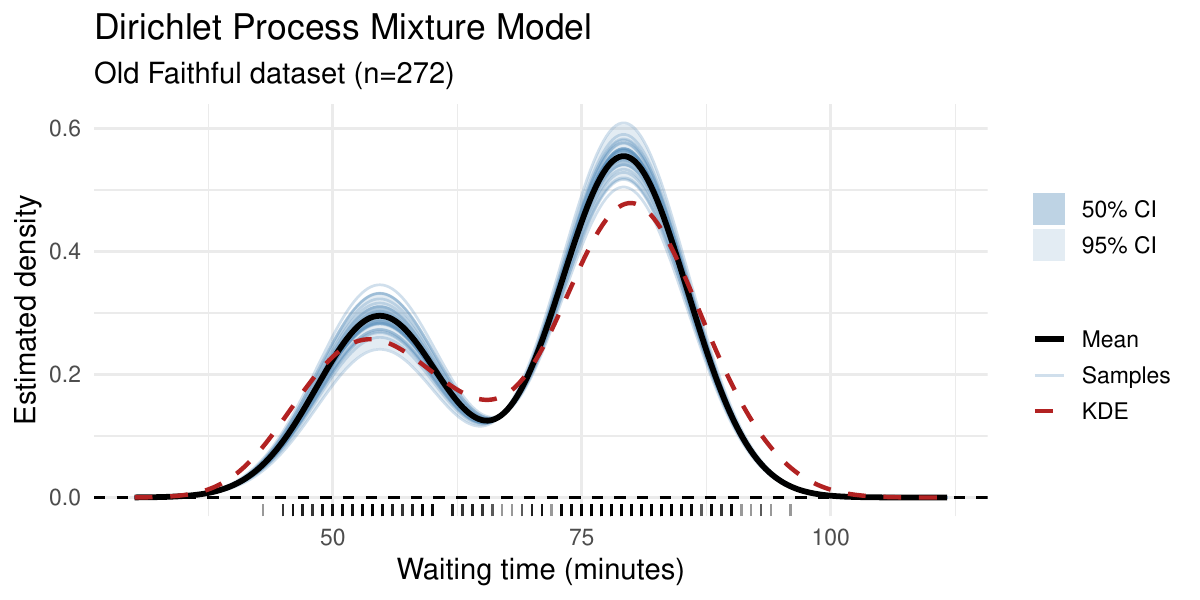}
    \caption{A DPMM applied to the old faithful dataset, annotations the same as in Figure \ref{fig:old_faith_MP_joint}}.
    \label{fig:old_faith_dpmm}
\end{figure}

We also empirically check convergence by computing a KDE every 1000 iterations of the predictive resampling scheme, and evaluating the densities at a fixed set of locations corresponding to quantiles of the original data. We do this for a subset of the sample paths, and display the results in Figure \ref{fig:old_faith_conv}, showing that after $N=100,000$ resampling steps, the density evaluations appear stable.

Finally, we applied the exact same procedure to the 'Galaxies' dataset \citep{roeder1990density}, consisting of $\nobs=82$ velocities of galaxies in the Corona Borealis region. We note similar behaviour as for the 'Old Faithful' dataset, as shown in Figure \ref{fig:galaxies_MP}. In particular, we note that the constant $c=1.5$ again brings the KDE of the original data close to the posterior mean, indicating perhaps that this value is stable across datasets, albeit possibly depending on the smoothness parameter $a$ in the bandwidth $h_n\propto n^{-a}$.

\begin{figure}
    \centering
    \includegraphics[width=0.85\textwidth]{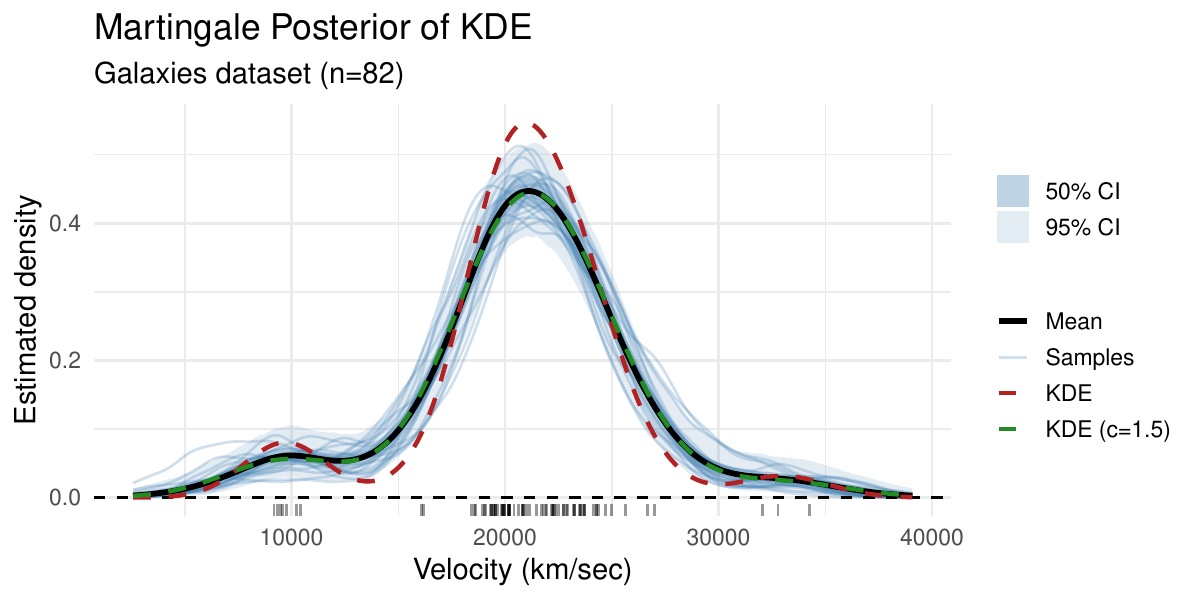}
    \caption{Martingale posterior of the KDE fit to the Galaxies dataset.}.
    \label{fig:galaxies_MP}
\end{figure}

\subsubsection*{Acknowledgements.} Parts of this work are based on Torjus Svardal{'}s bachelor{'}s thesis at the Department of Mathematics, University of Oslo. Leiv R{\o}nneberg and Emil A.~Stoltenberg acknowledge the support of {\it Integreat -- The Norwegian centre for knowledge-drive machine learning}, Norwegian Research Council Project number 332645~(01.07.2023-30.06.2033). Leiv R{\o}nneberg was further supported by the European
Union’s Horizon Europe research and innovation programme under the Marie Sklodowska-Curie grant
agreement No. 101126636. Many thanks to the participants in the predictive Bayes reading group at Blindern, and to Professor Nils Lid Hjort for helpful comments.


\newpage
\appendix

\section{Proofs left out of the main text}\label{app:proofs}
\subsection{Proof of Lemma~\ref{momentlemma}}\label{app:momentlemma}
\begin{proof}[Proof of Lemma~\ref{momentlemma}] We prove this by (strong) induction. Let $p \geq1$ and assume that $E\,|Y_1|^p <\infty$ and $\int |z|^p K(z) \,\dd z < \infty$. The expectation of $Y_2^p$ is 
\begin{align*}
\E\, Y_2^p  
& = \E\, \E\,\{ Y_2^p \given \calG_1\} 
= \E\, \int (h_1z + Y_1)^pK(z) \,\dd z
= \E\, \int \sum_{j=0}^p \binom{p}{j} (h_1z)^{p-j}Y_1^jK(z) \,\dd z\\
& = c_p h_1^p 
+ \sum_{j=1}^{p-1}\binom{p}{j} h_1^{p-j} c_{p-j}\E\,Y_1^j + \E\, Y_1^p ,
\end{align*}
where $c_k = \int z^k K(z)\,\dd z$. The expectation of $Y_3^p$ is
\begin{align*}
\E\, Y_3^p  
& = \E\, \E\,\{ Y_3^p\given \calG_2\} 
= \frac{1}{2}\sum_{i=1}^2 \E\, \int \sum_{j=0}^p \binom{p}{j} (h_2z)^{p-j}Y_i^jK(z) \,\dd z\\
& = 
c_p h_2^p + \frac{1}{2}\E\, Y_2^p + \frac{1}{2}\E\, Y_1^p 
+ \sum_{j=1}^{p-1}\binom{p}{j} h_2^{p-j} c_{p-j}\frac{1}{2}\sum_{i=1}^2\E\,Y_i^j\\
& = c_p h_2^p + \frac{1}{2}\bigg\{ c_p h_1^p + \E\, Y_1^p 
+ \sum_{j=1}^{p-1}\binom{p}{j} h_1^{p-j} c_{p-j}\E\,Y_1^j  \bigg\} + \frac{1}{2}\E\, Y_1^p 
+ \sum_{j=1}^{p-1}\binom{p}{j} h_2^{p-j} c_{p-j}\frac{1}{2}\sum_{i=1}^2\E\,Y_i^j\\
& = c_p h_2^p + c_p\frac{h_1^p}{2} + \E\, Y_1^p
+ \frac{1}{2}\sum_{j=1}^{p-1}\binom{p}{j} h_1^{p-j} c_{p-j}\E\,Y_1^j
+ \sum_{j=1}^{p-1}\binom{p}{j} h_2^{p-j} c_{p-j}\frac{1}{2}\sum_{i=1}^2\E\,Y_i^j.
\end{align*}
and, introducing $a_{m,j} = m^{-1}\sum_{i=1}^m \E\,Y_i^j$, we find that the expectation of $Y_4^p$ is
\begin{align*}
\E\, Y_4^p  
& = \E\, \E\,\{ Y_4^p \given\calG_3\} 
= \frac{1}{3}\sum_{i=1}^3 \E\, \int \sum_{j=0}^p \binom{p}{j} (h_3z)^{p-j}Y_i^jK(z) \,\dd z\\
& = c_ph_3^p + \frac{1}{3}\sum_{i=1}^3 \E\, Y_i^p  +\frac{1}{3}\sum_{i=1}^3 \E\, \int  \sum_{j=1}^{p-1} \binom{p}{j} (h_3z)^{p-j}Y_i^jK(z) \,\dd z \\
& = c_ph_3^p + c_p\frac{h_2^p}{3}
+ c_p\frac{h_1^p}{2} + \E\, Y_1^p
+ \frac{1}{2}\sum_{j=1}^{p-1}\binom{p}{j} h_1^{p-j} c_{p-j}\E\,Y_1^j\\
& \qquad \qquad + \frac{1}{3}\sum_{j=1}^{p-1}\binom{p}{j} h_2^{p-j} c_{p-j}\frac{1}{2}\sum_{i=1}^2\E\,Y_i^j 
+  \sum_{j=1}^{p-1} \binom{p}{j} h_3^{p-j}c_{p-j}\frac{1}{3}\sum_{i=1}^3\E\,Y_i^j \\
& = c_ph_3^p + c_p\frac{h_2^p}{3}
+ c_p\frac{h_1^p}{2} + \E\, Y_1^p
+ \frac{1}{2}\sum_{j=1}^{p-1}\binom{p}{j} h_1^{p-j} c_{p-j}a_{1,j}\\
& \qquad \qquad + \frac{1}{3}\sum_{j=1}^{p-1}\binom{p}{j} h_2^{p-j} c_{p-j}a_{2,j} 
+  \sum_{j=1}^{p-1} \binom{p}{j} h_3^{p-j}c_{p-j}a_{3,j}\\
& = c_ph_3^p + c_p\sum_{\ell=1}^2\frac{h_{\ell}^p }{\ell+1} 
+ \sum_{j=1}^{p-1} \binom{p}{j} \bigg( \sum_{\ell=1}^2\frac{h_{\ell}^{p-j} a_{\ell,j} }{\ell+1}\bigg) c_{p-j}
+ \sum_{j=1}^{p-1} \binom{p}{j} h_3^{p-j}c_{p-j}a_{3,j}
+ \E\, Y_1^p.
\end{align*}
This leads to the strong induction hypothesis that
\begin{align*}
\E\,Y_{m+1}^p    
= c_p h_m^p + c_p\sum_{\ell=1}^{m-1} \frac{h_{\ell}^p}{\ell + 1} 
+ \sum_{j=1}^{p-1} \binom{p}{j} h_m^{p-j}c_{p-j}a_{m,j}
+ \sum_{j=1}^{p-1} \binom{p}{j} c_{p-j}\bigg( \sum_{\ell=1}^{m-1}\frac{h_{\ell}^{p-j} a_{\ell,j} }{\ell+1}\bigg) 
+ \E\, Y_1^p,
\end{align*}
holds for $m = 1,2,3,\ldots,n$. We must show that it is true for $m = n+1$. First, 
\begin{equation}
\begin{split}
\E\,Y_{n+2}^p & = \frac{1}{n+1}\sum_{i=1}^{n+1}\E\,\int
\sum_{j=0}^p\binom{p}{j} (h_{n+1}z)^{p-j} Y_i^j K(z)\,\dd z\\
& = 
c_p h_{n+1}^p + \frac{1}{n+1}\sum_{i=2}^{n+1}\E\, Y_i^{p} + 
\sum_{j=1}^{p-1}\binom{p}{j} h_{n+1}^{p-j}c_{p-j}a_{n+1,j} + \E\,Y_1^p.
\end{split}
\label{eq:inductionstep}
\end{equation}
Here we need to check that the sum $(n+1)^{-1} \sum_{i=2}^{n+1}\E\, Y_i^p$ is of the hypothesised form. Inserting the induction hypothesis into this sum, we write
\begin{equation}
    \frac{1}{n+1} \sum_{i=2}^{n+1}\E\, Y_i^p  = A_{n+1} + B_{n+1} 
    + \frac{n}{n+1}\sum_{j=1}^{p-1}\binom{p}{j} h_{n+1}^{p-j}c_{p-j}a_{n+1,j}
    + \frac{n}{n+1}\E\,Y_1^p,  
    \label{eq:thesum}
\end{equation}
in terms of 
\begin{equation}
A_{n+1} = \frac{c_p}{n+1} \sum_{i=2}^{n+1} \bigg( h_{i-1}^p + \sum_{\ell = 1}^{i-2}\frac{h_{\ell}^p}{\ell + 1} \bigg),
    \notag
\end{equation}
and 
\begin{align*}
B_{n+1} & = \sum_{j=1}^{p-1}\binom{p}{j}\frac{c_{p-j}}{n+1}      
\sum_{i=2}^{n+1} \bigg( h_{i-1}^pa_{i-1,j} + \sum_{\ell = 1}^{i-2}\frac{h_{\ell}^pa_{\ell,j}}{\ell + 1} \bigg). 
\end{align*}
We see that these two sums contain) sums of the same form. In particular,
\begin{align*}
\frac{1}{n+1}      
\sum_{i=2}^{n+1} \bigg( h_{i-1}^pa_{i-1,j} + \sum_{\ell = 1}^{i-2}\frac{h_{\ell}^pa_{\ell,j}}{\ell + 1} \bigg) 
& = \frac{1}{n+1}\bigg( \sum_{\ell=1}^n h_{\ell}^pa_{\ell,j} + \sum_{\ell=1}^{n-1}\frac{n - \ell}{\ell + 1} h_{\ell}^p a_{\ell,j}\bigg)\\
& = \frac{1}{n+1}\sum_{\ell=1}^{n-1} h_{\ell}^p a_{\ell,j}\bigg(1 + \frac{n-\ell}{\ell + 1}\bigg) + \frac{h_n^{p}a_{n,j}}{n+1}\\
& = \sum_{\ell = 1}^{n}\frac{h_{\ell}^p a_{\ell,j}}{\ell + 1}.
\end{align*}
In summary, doing a similar computation for $A_{n+1}$, we have
\begin{equation}
A_{n+1} = c_p \sum_{\ell = 1}^{n}\frac{h_{\ell}^p }{\ell + 1},\quad \text{and}\quad 
B_{n+1} = \sum_{j=1}^{p-1}\binom{p}{j}c_{p-j}\sum_{\ell = 1}^{n}\frac{h_{\ell}^p a_{\ell,j}}{\ell + 1}.
    \notag
\end{equation}
Inserting~\eqref{eq:thesum} we these expressions for $A_{n+1}$ and $B_{n+1}$ back into the right hand side of~\eqref{eq:inductionstep} proves the induction hypothesis. It is tedious but easy to check that all expectations can be replaced by absolute moments $\E\,|Y_i|^j$ and $\int |z|^j K(z) \,\dd z$ at all steps in the above derivations, thus yielding an expression for $\E\,|Y_{n+1}|^p$.  
\end{proof}

\subsection{Proof of Theorem~\ref{thm:abscont}}\label{app:abscont}
We present this proof before that of Proposition~\ref{prop:notacid}, because the latter relies on the former. Let the sequence of predictives $(\alpha_n)_{n\geq 0}$ be as defined in~\eqref{eq:predkernel1}, with Gaussian kernel $K(z) = \phi(z) = (2\pi)^{-1/2}\exp(-\tfrac12 z^2)$ and bandwidth $h_n = cn^{-a}$. Write $\Phi(z) = \int_{-\infty}^z \phi(u)\,\dd u$ for the standard Gaussian cumulative. We begin with a purely algebraic lemma.

\begin{lem}\label{lem:integral}
Let $\mu_1, \mu_2 \in \real$ and $\sigma_1, \sigma_2 > 0$. Then
\begin{equation}
    \int_{-\infty}^{\infty} \frac{1}{\sigma_1}\phi\left(\frac{y - \mu_1}{\sigma_1}\right)\frac{1}{\sigma_2}\phi\left(\frac{y - \mu_2}{\sigma_2}\right)\,\dd y = \frac{1}{\sqrt{\sigma_1^2 + \sigma_2^2}}\phi\left(\frac{\mu_1 - \mu_2}{\sqrt{\sigma_1^2 + \sigma_2^2}}\right),
    \notag
\end{equation}

\end{lem}
\begin{proof}
    First make the substitution $x = (y - \mu_1)/\sigma_1$ to obtain
    $$\int\phi\left(\frac{y - \mu_1}{\sigma_1}\right)\phi\left(\frac{y - \mu_2}{\sigma_2}\right)\,\dd y = \sigma_1\int\phi(x)\phi\left(\frac{\mu_1 - \mu_2}{\sigma_2} + \frac{\sigma_1}{\sigma_2}x\right)\,\dd x,$$
    and then apply the identity (see \citet[p.~396, Table 1, entry n10]{owen1980table})
    \begin{equation}\int\phi(x)\phi(a + bx)\,\dd x = \frac1{\sqrt{1 + b^2}}\phi\left(\frac{a}{\sqrt{1 + b^2}}\right)\Phi\left(\sqrt{1 + b^2}x + \frac{ab}{\sqrt{1 + b^2}}\right) + C,
    \notag
    \end{equation}
    with $a = (\mu_1 - \mu_2)/\sigma_2$ and $b = \sigma_1/\sigma_2$.\end{proof}
Next, we need to show that a certain stochastic process is \as~bounded.
\begin{lem}\label{lem:LLN}
    Let $Y_1, Y_2, \ldots$ be distributed according to the predictives $\alpha_0, \alpha_1, \ldots$ Then
    \begin{equation*}\frac1{n^2h_n}\sum_{i, j = 1}^n\phi\left(\frac{Y_i - Y_j}{h_n}\right) \to Z\quad\as,
    \end{equation*}
    where $|Z| < \infty$ a.s. In particular,
    \begin{equation*}
        \sup_{n\geq 1}\frac{1}{n^2h_n}\sum_{i,j=1}^n\phi\left(\frac{Y_i - Y_j}{h_n}\right) < \infty\quad\as
    \end{equation*}
\end{lem}

\begin{proof}
Let 
\begin{equation*}
    Q_n = \frac1{n^2h_n}\sum_{i,j=1}^n\phi\left(\frac{Y_i - Y_j}{h_n}\right).
\end{equation*}
Then
    \begin{multline}\label{eq:three_terms}
        \E\left\{Q_{n+1}\given\calG_n\right\} = \frac1{(n+1)^2h_{n+1}}\sum_{i,j=1}^n\phi\left(\frac{Y_i - Y_j}{h_{n+1}}\right) \\ + \frac2{(n+1)^2h_{n+1}}\sum_{j=1}^n\E\left\{\phi\left(\frac{Y_{n+1} - Y_j}{h_{n+1}}\right)\,\middle|\,\calG_n\right\} + \frac1{\sqrt{2\pi}}\frac1{(n+1)^2h_{n+1}}.
    \end{multline}
We require an upper bound for the middle term of~\eqref{eq:three_terms}. As $Y_{n+1}\given\calG_n \sim \sum_{i=1}^nn^{-1}\mathrm{N}(Y_i, h_n^2)$, we have that 
\begin{multline*}
    \E\left\{\phi\left(\frac{Y_{n+1} - Y_j}{h_{n+1}}\right)\,\middle|\,\calG_n\right\} = \frac1n\sum_{i=1}^n\int\phi\left(\frac{y - Y_j}{h_{n+1}}\right)\phi\left(\frac{y - Y_i}{h_n}\right)\,\dd y \\ = \frac{h_{n+1}}{n\sqrt{h_n^2 + h_{n+1}^2}}\sum_{i=1}^n\phi\left(\frac{Y_i - Y_j}{\sqrt{h_n^2 + h_{n+1}^2}}\right),
\end{multline*}
where we have used Lemma~\ref{lem:integral}. Hence 
\begin{align*}
    \frac2{(n+1)^2h_{n+1}}\sum_{j=1}^n\E\left\{\phi\left(\frac{Y_{n+1} - Y_j}{h_{n+1}}\right)\,\middle|\,\calG_n\right\} & = \frac2{n(n+1)^2\sqrt{h_n^2 + h_{n+1}^2}}\sum_{i,j=1}^n\phi\left(\frac{Y_i - Y_j}{\sqrt{h_n^2 + h_{n+1}^2}}\right) \\
    & \leq \frac{2}{n^3 h_n}\sum_{i,j=1}^n\phi\left(\frac{Y_i - Y_j}{\sqrt{2}h_n}\right),
\end{align*}
where we have used that
\begin{equation*}
    \phi\left(\frac{x}{h_n}\right) \leq \phi\left(\frac{x}{\sqrt{h_n^2 + h_{n+1}^2}}\right) \leq \phi\left(\frac{x}{\sqrt{2}h_n}\right)
\end{equation*}
for all $x\in\mathbb{R}$.

Let $\delta_n \geq 0$ be a sequence, to be chosen shortly. If $|x| \leq \delta_n h_n$, then
\begin{equation*}
    \phi\left(\frac{x}{\sqrt{2}h_n}\right) = \exp\left\{\frac14\left(\frac{x}{h_n}\right)^2\right\} \phi\left(\frac{x}{h_n}\right) \leq e^{\delta_n^2 / 4}\phi\left(\frac{x}{h_n}\right),
\end{equation*}
so
\begin{align*}
    \sum_{i,j=1}^n\phi\left(\frac{Y_i - Y_j}{\sqrt{h_n}}\right) & = \sum_{|Y_i - Y_j| \leq \delta_nh_n}\phi\left(\frac{Y_i - Y_j}{\sqrt{2}h_n}\right) + \sum_{|Y_i - Y_j| > \delta_nh_n}\phi\left(\frac{Y_i - Y_j}{\sqrt{2}h_n}\right) \\
    & \leq e^{\delta_n^2 / 4}\sum_{|Y_i - Y_j| \leq \delta_nh_n} \phi\left(\frac{Y_i - Y_j}{h_n}\right) + n^2\phi\left(\frac{\delta_n}{\sqrt{2}}\right) \\
    & \leq e^{\delta_n^2 / 4}\sum_{i,j=1}^n\phi\left(\frac{Y_i - Y_j}{h_n}\right) + n^2\phi\left(\frac{\delta_n}{\sqrt{2}}\right).
\end{align*}
Now, let $\delta_n = 2\sqrt{\log n}$. Then $e^{\delta_n^2 / 4} = n$, and $\phi(\delta_n / \sqrt{2}) = (2\pi)^{-1/2}\exp\left\{-\frac14\log n\right\} = (2\pi)^{-1/2}n^{-1}$, and so
\begin{align*}
    \frac2{n^3h_n}\sum_{i,j=1}^n\phi\left(\frac{Y_i - Y_j}{\sqrt{2}h_n}\right) & \leq \frac{2}{n^2 h_n}\sum_{i,j=1}^n\phi\left(\frac{Y_i - Y_j}{h_n}\right) + \frac{2}{\sqrt{2\pi}}\frac1{n^2 h_n} \\
    & = 2Q_n + \sqrt{\frac{2}{\pi}}\frac1{n^2 h_n}.
\end{align*}
Combining this bound with the fact that $\phi(x/h_{n+1}) \leq \phi(x/h_n)$ for all $x\in\mathbb{R}$, we have thus established that
\begin{equation*}
    \E\{Q_{n+1}\given\calG_n\} \leq \left\{2 + \left(\frac{n}{n+1}\right)^2\frac{h_n}{h_{n+1}}\right\} Q_n + \sqrt{\frac{2}{\pi}}\frac1{n^2h_n} + \frac1{\sqrt{2\pi}}\frac1{(n+1)^2 h_{n+1}}\quad\as,
\end{equation*}
and so, since $h_n / h_{n+1} \to 1$ and $\sum_{n=1}^\infty(n^2 h_n)^{-1} \propto \sum_{n=1}^\infty n^{a - 2} < \infty$ provided $a < 1$, we may apply the convergence theorem for nonnegative almost supermartingales by \citet{robbins1971convergence}, which proves the result.
\end{proof}

Before proving absolute continuity of the limiting distribution, we require one final lemma, which is a brief extension of the first implication in de la Vall\'{e}e-Poussin's theorem (see for example \citet[Theorem~4.5.9, p.~272]{bogachev2007measureI}) to measures with values in $[0,\infty]$. Let $(\mathcal{X}, \calF, \mu)$ be a measure space (not necessarily finite). Using Definition 4.5.2 from \citet[p.~267]{bogachev2007measureI}, we say that a family $\mathcal{A}$ of real-valued measurable functions on $\mathcal{X}$ has \emph{uniformly absolutely continuous integrals} if
\begin{equation}
    \lim_{\mu(E)\to 0}\sup_{f\in\mathcal{A}}\int_E|f|\,\dd\mu = 0.
\end{equation}

\begin{lem}\label{lem:delavalleepoussin}
    Let $(\mathcal{X}, \calF, \mu)$ be as above and let $\mathcal{A}$ be a family of real-valued measurable functions on $\mathcal{X}$. Suppose there exists a strictly increasing function $G : [0, \infty) \to [0, \infty)$ such that
    \begin{equation}
        \lim_{t\to\infty}\frac{G(t)}{t} = \infty\quad\text{and}\quad\sup_{f\in\mathcal{A}}\int G(|f|)\,\dd\mu = M < \infty.\notag
    \end{equation}
    Then $\mathcal{A}$ has uniformly absolutely continuous integrals.
\end{lem}

\begin{proof}
    Let $\varepsilon > 0$ and take $c \in \mathbb{R}$, yet to be specified. We have
    \begin{equation}
        \sup_{f\in\mathcal{A}} \int_E |f|\dd\mu \leq \sup_{f\in\mathcal{A}}\int_E|f| I_{|f| \leq c}\,\dd\mu + \sup_{f\in\mathcal{A}}\int_E |f| I_{|f| > c}\,\dd \mu.
        \notag
    \end{equation}
    Now, the first term is bounded above by $c\mu(E)$. As for the second,
    \begin{equation}
    \sup_{f\in\mathcal{A}}\int_E |f| I_{|f| > c}\,\dd\mu = \sup_{f\in\mathcal{A}}\int_E\frac{|f|}{G(|f|)}G(|f|)I_{|f| > c}\,\dd\mu \leq \frac{cM}{G(c)},
    \notag
    \end{equation}
    so if we take $c$ large enough so that $G(c) / c > M/\varepsilon$, and $\mu(E) < \varepsilon / c$, then we obtain
    \begin{equation}
        \sup_{f\in\mathcal{A}}\int_E|f|\,\dd\mu < 2\varepsilon,\notag
    \end{equation}
    proving the lemma.
\end{proof}

We are now ready to prove the main result of this section.
\begin{proof}[Proof of Theorem~\ref{thm:abscont}]
    Let 
    \begin{equation}
        \mathcal{R} = \left\{\omega\in\Omega : \sup_{n\geq 1}\frac1{n^2 h_n}\sum_{i,j=1}\phi\left(\frac{Y_i(\omega) - Y_j(\omega)}{h_n}\right) < \infty\right\}\notag
    \end{equation}
    and
    \begin{equation}
        \mathcal{S} = \left\{\omega\in\Omega : \int g(y)\alpha_n(\dd y, \omega) \to \int g(y)\alpha(\dd y, \omega)\text{ for all }g\in C_b\right\},\notag
    \end{equation}
    where we recall that $C_b$ denotes the set of bounded and continuous functions from $\mathbb{R}$ to itself. By Lemma~\ref{lem:LLN} and Theorem~\ref{thm:weakconvas}, respectively, we know that $\Pr(\mathcal{R}) = \Pr(\mathcal{S}) = 1$, and so in particular, $\Pr(\mathcal{R}\cap\mathcal{S}) = 1$. 

    Take $\omega \in \mathcal{R}\cap\mathcal{S}$ and let $G(t) = t^2$. By Lemmas~\ref{lem:integral} and~\ref{lem:LLN},
\begin{multline*}
    \sup_{n\geq 1}\int G(|f_n(y, \omega)|)\,\dd y = \sup_{n\geq 1}\frac1{n^2h_n^2}\sum_{i,j=1}^n\int\phi\left(\frac{y - Y_i(\omega)}{h_n}\right)\phi\left(\frac{y - Y_j(\omega)}{h_n}\right)\,\dd y \\ = \sup_{n\geq 1}\frac1{n^2h_n}\sum_{i,j=1}^n\phi\left(\frac{Y_i(\omega) - Y_j(\omega)}{h_n}\right) < \infty,
\end{multline*}
as $\omega\in\mathcal{R}$, and so by Lemma~\ref{lem:delavalleepoussin}, the sequence $(f_n(\cdot, \omega))_{n\geq 1}$ has uniformly absolutely continuous integrals. Furthermore, as $\omega\in\mathcal{S}$, the sequence $(\alpha_n(\cdot, \omega))_{n\geq 0}$ is tight by Prokhorov's theorem (see, e.g., \citet[Theorem~6.2, p.~37]{Billingsley68}). That is, for any given $\varepsilon > 0$, there exists a compact set $K_{\eps, \omega}$ such that 
\begin{equation}
\sup_{n\geq 1}\int_{K_{\eps}(\omega)^c}f_n(y, \omega)\,\dd y
= \sup_{n\geq 1} \alpha_n(K_{\eps, \omega}^c,\omega) < \eps.
\notag
\end{equation}
Finally, we note that $\lVert f_n(\cdot, \omega)\rVert_1 = 1$ for all $n\geq 1$. Hence, by Theorem 4.7.20 in~\citet[p.~287]{bogachev2007measureI}, the sequence $(f_n(\cdot, \omega))_{n\geq 1}$ is relatively compact in the weak topology of $L^1(\mathbb{R})$. Hence, there exists $f(\cdot, \omega)\in L^1(\mathbb{R})$ and a subsequence $n_k(\omega)$ such that $f_{n_k(\omega)}(\cdot, \omega) \Rightarrow f(\cdot, \omega)$, that is,
\begin{equation}
    \int g(y) f_{n_k(\omega)}(y, \omega)\,\dd y \to \int g(y) f(y, \omega)\,\dd y\text{ for all }g\in L^{\infty}(\mathbb{R}).\notag
\end{equation}
In particular, this convergence holds for all $g \in C_b$. As $\omega\in\mathcal{S}$, we also have
\begin{equation}
    \int g(y) \alpha_{n_k(\omega)}(\dd y, \omega) \to \int g(y)\alpha(\dd y, \omega)\text{ for all }g\in C_b,
    \notag
\end{equation}
but by definition,
\begin{equation}
    \int g(y) \alpha_{n_k(\omega)}(\dd y, \omega) = \int g(y) f_{n_k(\omega)}(y, \omega)\,\dd y\text{ for all }g\in C_b,
    \notag
\end{equation}
which forces
\begin{equation}
    \int g(y) f(y, \omega)\,\dd y = \int g(y) \alpha(\dd y, \omega)\text{ for all }g\in C_b,
    \notag
\end{equation}
by uniqueness of limits. As $C_b$ separates measures on $(\mathbb{R}, \calB)$, $f(\cdot, \omega)$ must be the density of the distribution $\alpha(\cdot, \omega)$. As the argument applies to any $\omega$ from a set of measure 1, we are done.
\end{proof}

\subsection{Proof of Proposition~\ref{prop:notacid}}\label{app:notacid}
\begin{proof}[Proof of Proposition~\ref{prop:notacid}]
By Theorem~\ref{thm:weakconvas}, we know there exists a distribution $\alpha$ such that $\alpha_n\Rightarrow\alpha~\as$ Let $A$ be a set of the form $(b, \infty)$ such that $\alpha(A) > 0$. We have that 
\begin{align*}
    \E\,\{\alpha_{n+1}(A)\given\calG_n\} & = \E\,\bigg\{\frac1{(n+1)h_{n+1}}\sum_{i=1}^{n+1}\int_A\phi\left(\frac{y-Y_i}{h_{n+1}}\right)\,\dd y\given\calG_n\bigg\} \\
    & = \frac1{(n+1)h_{n+1}}\sum_{i=1}^n\int_A\phi\left(\frac{y - Y_i}{h_{n+1}}\right)\,\dd y\\ 
    & \quad \qquad + \frac1{(n+1)h_{n+1}}\E\,\bigg\{\int_A \phi\left(\frac{y - Y_{n+1}}{h_{n+1}}\right)\,\dd y\given \calG_n\bigg\} \\
    & = \frac1{(n+1)h_{n+1}}\sum_{i=1}^n\int_A\phi\left(\frac{y - Y_i}{h_{n+1}}\right)\,\dd y \\
    & \quad\qquad + \frac1{(n+1)h_{n+1}}\int_{A}\int\phi\left(\frac{y - y_{n+1}}{h_{n+1}}\right)\frac1{n h_n}\sum_{i=1}^n\phi\left(\frac{y_{n+1} - Y_i}{h_n}\right)\,\dd y_{n+1}\,\dd y
\end{align*}
\as, where the final integrals can be exchanged by Tonelli's theorem. Now, using Lemma~\ref{lem:integral}, we get
\begin{align*}
\E\,\{\alpha_{n+1}(A)\given\calG_n\} & = \frac1{(n+1)h_{n+1}}\sum_{i=1}^n\int_A\phi\left(\frac{y - Y_i}{h_{n+1}}\right)\,\dd y\\  
& \qquad \qquad + \frac{1}{n(n+1)}\sum_{i=1}^n \int_A \frac{1}{\sqrt{h_{n+1}^2 + h_n^2}}\phi\left(\frac{y - Y_i}{\sqrt{h_{n+1}^2 + h_n^2}}\right) \,\dd y,
\end{align*}
and so
\begin{equation}
\begin{multlined}[c]
\E\,\{\alpha_{n+1}(A)\given\calG_n\} - \alpha_n(A) = \alpha_{n+1}(A) - \alpha_n(A) - \frac1{(n+1)h_{n+1}}\int_A \phi\bigg(\frac{y - Y_{n+1}}{h_{n+1}}\bigg)\,\dd y \\ \quad + \frac{1}{n(n+1)}\sum_{i=1}^n \int_A \frac{1}{\sqrt{h_{n+1}^2 + h_n^2}}\phi\left( \frac{y - Y_i}{\sqrt{h_{n+1}^2 + h_n^2}}\right) \,\dd y\quad\as
\end{multlined}
\label{eq:not-acid-three-terms}
\end{equation}
By Theorems~\ref{thm:weakconvas} and~\ref{thm:abscont}, we know that $\alpha_n(A) \to \alpha(A)$ (as $A$ is a continuity set), and so the difference $\alpha_{n+1}(A) - \alpha_n(A)$ converges to 0 a.s. As for the next term of~\eqref{eq:not-acid-three-terms},
\begin{equation}
0 \leq \frac1{(n+1)h_{n+1}}\int_A \phi\bigg(\frac{y - Y_{n+1}}{h_{n+1}}\bigg)\,\dd y =  \frac1{n+1}\Pr(Y\in A\given Y_{n+1}) \leq \frac1{n+1},
    \notag
\end{equation}
where $Y\given Y_{n+1} \sim {\rm N}(Y_{n+1}, h_{n+1}^2)$, and so this also converges to 0 a.s. To analyse the final term, note that for all $i\geq 1$ and for any $\sigma > 0$,
\begin{equation}
\frac{1}{\sigma} \int_b^\infty\phi\bigg(\frac{y - Y_i}{\sigma}\bigg)\,\dd y = \Phi\left(\frac{Y_i-b}{\sigma}\right) \geq \half I_{Y_i > b}\quad{\as}
    \notag
\end{equation}
Hence, with $A = (b,\infty)$, we have the lower bound
\begin{equation}
\frac{1}{n(n+1)}\sum_{i=1}^n \int_A \frac{1}{\sqrt{h_{n+1}^2 + h_n^2}}\phi\left(\frac{y - Y_i}{\sqrt{h_{n+1}^2 + h_n^2}}\right) \,\dd y
\geq \frac{1}{2n(n+1)}\sum_{i=1}^nI_{Y_i > b}\quad{\as}.
    \notag
\end{equation}
Now, by Proposition~\ref{prop:empmeasure_conv}, we have that 
\begin{equation}
\frac{1}{n}\sum_{i=1}^n I_{Y_i > b} \to \alpha(A) > 0 \quad {\as},
    \notag
\end{equation}
and so for sufficiently large $n$, we have
\begin{equation}
\E\,\{\alpha_{n+1}(A)\given\calG_n\} - \alpha_n(A) \geq \frac12\frac{\alpha(A)}{n+1} + C, \quad {\as},
    \notag
\end{equation}
for some constant $C$ (which takes into account the other terms of~\eqref{eq:not-acid-three-terms}). Hence, because $\sum_{n=1}^\infty 1/(n+1) = \infty$, the sequence $(\alpha_n(A))$ is not a nonnegative almost supermartingale for this choice of set $A$. Since the a.c.i.d.~condition requires it to be so for all sets $A$, the proof is completed.
\end{proof}

\section{Some weak convergence results}\label{moreweak}
Let $(Y_n)_{n \geq 1}$, with $\alpha_n(A) = \Pr(Y_{n+1}\in A \given \calG_n),\, A \in \calB$, where $\calG_0 = \{\emptyset, \Omega\}$ and $\calG_n = \sigma(Y_1,\ldots,Y_n)$ for $n \geq 1$. According to Theorem~2.2 in \citet{berti2006almost} we have $\alpha_n \Rightarrow \alpha$ {\as}~if and only if $\E\,\{g(Y_{n+1}) \given \calG_n\}$ is {\as}~convergent for all $g \in C_b$. The fact that the reals equipped with the Euclidean metric is a Polish metric space is important for this equivalence. In this appendix we show that such an equivalence holds when $g$ is restricted to be nonnegative and smooth. In proving this we rely on Theorem~2.6 in \citet{berti2006almost}, which says that $\alpha_n \Rightarrow \alpha$ {\as}~if and only the conditional characteristic functions converge {\as}, that is
\begin{equation}
\int \exp(ity) \,\alpha_n(\dd y)\to 
\int \exp(ity) \,\alpha(\dd y),\quad \as.
    \notag
\end{equation}
For how {\as}~weak convergence of random probability measures relates to stable convergence and mixing convergence \citep{aldous1978mixing,jacod2003limit,hausler2015stable}, see \citet[Sect.~4]{grubel2016functional}. 

We also remark that L{\'e}vy{'}s continuity theorem (see, e.g., Corollary~2 in \citet[p.~350]{billingsley96pm}) carries over to random probability measures.

\begin{lem} If $\int \exp(ity) \,\alpha_n(\dd y,\omega)\to\beta(t,\omega)$ for all $\omega \in N^c$ where $\Pr(N) = 0$, and $t \mapsto \beta(t,\omega)$ is continuous in $t = 0$ for all $\omega \in N^c$, $\beta(t,\cdot)$ is the (random) characteristic function of some random probability measure $\alpha$.   
\end{lem}
\begin{proof} Let $\varphi_n(t,\omega) = \int \exp(ity) \,\alpha_n(\dd y,\omega)$ for each $n$ and $\omega$. By the standard proof we get that for any $\delta > 0$
\begin{equation}
\alpha_n(\{y \colon |y|\geq 2/\delta\},\omega) \leq \frac{1}{\delta}\int_{-\delta}^{\delta}\{1 - \varphi_n(t,\omega)\}\,\dd t. 
    \notag
\end{equation}
By bounded convergence $\int_{-\eps}^{\eps}\{1 - \varphi_n(t,\omega)\}\,\dd t  \to \int_{-\eps}^{\eps}\{1 - \beta(t,\omega)\}\,\dd t$, for each $\omega \in N^c$. Given $\eps > 0$ we can then choose $\delta$ so that  $|t|< \delta$ we have $|1 - \beta(t,\omega)| < \eps$, and thus for each $\omega \in N^c$
\begin{equation}
\limsup_n  \alpha_n(\{y \colon |y|\geq 2/\delta\},\omega) \leq 2 \eps,
    \notag
\end{equation}
where $\eps = \eps(\omega)$. This shows for each $\omega \in N^c$, the sequence $(\alpha_n(\cdot,\omega))_{n \geq 0}$ of probability measures is tight, and so $\beta(t,\omega)$ is a the characteristic function of some probability measure for each $\omega \in N^c$, where $\Pr(N) = 0$.   
\end{proof}

\begin{lem}\label{lemma:aux1} Let $(Y_n)_{n \geq 1}$, $(\alpha_n)_{n \geq 0}$ and $(\calG_n)_{n \geq 0}$ be as described in the above paragraph. Take $Z_n \sim \normal(0,1),\,n\geq 1$, such that $(Z_n)_{n \geq 1}$ and $(Y_n)_{n \geq 1}$ are independent, and set $\mu_{\sigma,n}(A) = \Pr (Y_{n+1} + \sigma Z_{n+1} \in A \given \calG_n),\, A\in \calB$ for $n \geq 0$. Then $\alpha_n \Rightarrow \alpha$ {\as}~if and only if $\mu_{\sigma,n} \Rightarrow \mu_{\sigma}$ {\as}~for each $\sigma > 0$. 
\end{lem}
\begin{proof} The characteristic function of $Y_{n+1} + \sigma Z_{n+1}$ given $\calG_n$ is 
\begin{equation}
\varphi_{\sigma,n}(t) = \exp(-\half t^2 \sigma^2) \int \exp(it y)\,\alpha_n(\dd y). 
    \notag
\end{equation}
If $\alpha_n \Rightarrow \alpha$ {\as}, then $\varphi_{\sigma,n}(t) \to \exp(- \half t^2 \sigma^2) \int \exp(ity) \,\alpha(\dd y)$ {\as}. The product of two characteristic functions is a characteristic function, so by Theorem~2.6 in \citet{berti2006almost} we conclude that $\mu_{\sigma,n}\Rightarrow \mu_{\sigma}$ {\as}, where $\mu_{\sigma}$ is a random measure such that $\int \exp(itx)\,\mu_{\sigma}(\dd x) = \exp(- \half t^2 \sigma^2) \int \exp(ity) \,\alpha(\dd y)$ {\as}.

Conversely, assume that $\mu_{\sigma,n} \Rightarrow \mu_{\sigma}$ {\as} for each $\sigma > 0$, for some random probability measure $\mu_{\sigma}$. By the theorem just quoted, this implies that $\varphi_{\sigma,n}(t) \to \varphi_{\sigma}(t)$ {\as}~for all $t \in \real$, and for each $\sigma > 0$; where $\varphi_{\sigma}$ is the characteristic function of $\mu_{\sigma}$. We then have that
\begin{equation}
\int \exp(ity) \,\alpha_n(\dd y)
= \exp(\half t^2 \sigma^2) \varphi_{\sigma,n}(t) \to \exp(\half t^2 \sigma^2) \varphi_{\sigma}(t),\quad \as. 
\notag
\end{equation}
The function $t \mapsto \exp(\half  t^2\sigma^2 )\varphi_{\sigma}(t,\omega)$ is continuous in zero, for almost all $\omega$. By L{\'e}vy{'}s continuity theorem, the limit $\exp(\half  t^2\sigma^2 )\varphi_{\sigma}(t)$ must therefore be a characteristic function {\as} (equivalently, the sequence $(\alpha_n)_{n \geq 0}$ is {\as}~tight). Since $\sigma > 0$ was arbitrary and the limit of $\int \exp(it y)\,\alpha_n(\dd y)$ is unique (in the {\as}~sense), the limiting characteristic function cannot depend on $\sigma$. Consequently, there exists a random probability measure $\alpha$ such that $\int \exp(ity) \,\alpha_n(\dd y) \to \int \exp(ity) \,\alpha(\dd y)$ {\as}. By Theorem~2.6 in~\citet{berti2006almost} again, we can conclude that $\alpha_n \Rightarrow \alpha$ {\as}.
\end{proof}
The introduce-Gaussian-noise {`}trick{'} used in the next lemma is drawn from \citet[Section~III.5]{pollard1984convergence}, and is also included as an exercise in~\citet[p.~166]{jacod2004probability}. 
\begin{lem}\label{lemma:weakconvCbinfty} The conditions in~\eqref{eq:berti1} and \eqref{eq:ourcond} are equivalent.
\end{lem}
\begin{proof} Since $C_b^{\infty} \subset C_b$, \eqref{eq:berti1} implies \eqref{eq:ourcond}. Conversely, assume that \eqref{eq:ourcond} holds, and let $f \in C_b$. Take $Z_{n} \sim \normal(0,1)$ independent of $(Y_n)_{n \geq 1}$ for each $n \geq 1$. Then, for $\sigma > 0$
\begin{equation}
\E\,\{ f(Y_{n+1} + \sigma Z_{n+1} ) \given \calG_{n} \}
= \E\,\{ g_{\sigma}(Y_{n+1}) \given \calG_{n}\},\quad \text{a.s.}, 
    \notag
\end{equation}
where
\begin{equation}
g_{\sigma}(y) = \frac{1}{\sqrt{2\pi}\sigma}\int f(u)\exp\{-\half(y - u)^2/\sigma^2 \}\,\dd u,  
    \notag
\end{equation}
and one may check that $g_{\sigma} \in C_b^{\infty}$ as desired. Then, since we are assuming~\eqref{eq:ourcond}, the above {\as}~equality implies that $\E\, \{f(Y_{n+1} + \sigma Z_n)\given \calG_n\}$ is {\as}~convergent, and it follows from Theorem~2.2 in \citet{berti2006almost} that $\mu_{\sigma,n} \Rightarrow \mu_{\sigma}$. Since $\sigma > 0$ was arbitrary, we have from Lemma~\ref{lemma:aux1} that this is equivalent to $\alpha_n \Rightarrow \alpha$ {\as}.
\end{proof}
For a sequence $(Y_n)_{n \geq 1}$ let $\PP_n = n^{-1}\sum_{i=1}^n \delta_{Y_i}$ be the empirical measure of the $n$ first random variables. If the sequence $(Y_n)_{n \geq 1}$ is exchangeable, it is well known that the empirical measure converges weakly {\as}~to the directing measure~\citep[p.~14]{aldous1985exchangeability}. In fact, as is also known, asymptotic exchangeability is sufficient. The following proposition is included completeness. 
\begin{prop}\label{prop:empmeasure_conv} Let $(Y_n)_{n\geq 1}$ be a sequence of random variables with predictives $(\alpha_n)_{n \geq 0}$ such that $\alpha_n \Rightarrow \alpha$ {\as}. Then $\PP_n \Rightarrow \alpha$ {\as}.  
\end{prop}
\begin{proof} Let $A$ be a continuity set of $\alpha$, i.e., $\alpha(\partial A) = 0$ {\as}. Write
\begin{equation}
\PP_n(A) = \frac{1}{n}\sum_{k=1}\alpha_{k-1}(A) + \frac{1}{n}\sum_{k=1}^n k\frac{\delta_{Y_k}(A) - \alpha_{k-1}(A)}{k}.
    \notag
\end{equation}
Since $\alpha_n(A) \to \alpha(A)$ {\as}~by assumption, we have $n^{-1}\sum_{k=1}^n\alpha_{k-1}(A) \to \alpha(A)$ {\as}~by the Toeplitz lemma. Almost sure convergence of the second sum, to zero, follows from Kronecker{'}s lemma upon noting that $n \mapsto \sum_{k=1}^n k^{-1}\{\delta_{Y_k}(A) - \alpha_{k-1}(A)\}$ is a uniformly integrable martingale (see the proof of Lemma~\ref{lemma:meanconv}). This shows that $\PP_n(A) \to \alpha(A)$ {\as} for every continuity set $A$ of $\alpha$. Since $\real$ is Polish (ensuring that the null sets do not pile up), this implies that $\PP_n \Rightarrow \alpha$ {\as}. 
\end{proof}

\section{On criteria for a.s.~weak convergence}\label{app:convcriteria}
As touched upon in Section~\ref{sec:asweakconv} (below Theorem~\ref{thm:weakconvas}), there are three known general criteria ensuring that a sequence of predicitive distributions converge weakly~{\as} These three criteria are, one, that the $(\alpha_n)_{n \geq 0}$ form a {\it martingale} in the sense that
\begin{equation}
\E\,\{ \alpha_{n+1}(A) \given \calG_n\} = \alpha_n(A) \quad {\as},
    \notag
\end{equation}
for all $A \in \calB$; which is equivalent to the sequence $(Y_n)_{n \geq 1}$ being {\it conditionally identically distributed} (c.i.d.), this being defined as $\Pr(Y_{n+k} \in A \given \calG_n) = \alpha_n(A)$, for all $A \in \calB$ and all $k \geq 1$. Two, that the sequence $(\alpha_n)_{n \geq 0}$ forms a {\it quasi-martingale}, that is
\begin{equation}
\sum_{n \geq 0} \E\,\big| \E\, \{\alpha_{n+1}(A) \given \calG_n\} - \alpha_n(A)      \big| < \infty, 
    \label{eq:defquasi}
\end{equation}
for each $A \in \calB$ \citep[Def.~G.2]{fortini2026principled}. And, three, that the sequence $(Y_n)_{n \geq 1}$ is {\it almost} c.i.d., or {\it a.c.i.d.}, meaning that 
\begin{equation}
\E\,\{\alpha_{n+1}(A) \given \calG_n\} \leq \alpha_n(A) + \xi_n \quad {\as}
    \label{eq:defacid}
\end{equation}
for all $A \in \calB$, for a nonnegative {\as} summable sequence $(\xi_n)_{n\geq 0}$ \citep[Theorem~1]{battiston2025bayesian}.\footnote{Note that \citet{battiston2025bayesian} do not include the absolute summability condition as a defining property of a sequence being a.c.i.d. Hence our terminology differs somewhat from theirs.} The next well known lemma entails that if $(\alpha_n)_{n \geq 1}$ is a quasi-martingale, then it is a.c.i.d.

\begin{lem} Let $(X_n)_{n \geq 1}$ be nonnegative random variables such that $\sum_{n \geq 1} \E\, X_n < \infty$. Then $\sum_{n \geq 1} X_n < \infty$ {\as}.
\end{lem}
\begin{proof} Let $S_n = \sum_{j=1}^n X_j$, so that $(S_n)_{n \geq 1}$ is an increasing sequence of random variables. For $n > m$, we have $\E\,|S_n - S_m| = \sum_{j=m+1}^n \E\, X_j$, which tends to zero as $n,m \to \infty$, by the assumption. This shows that $(S_n)_{n\geq 1}$ is Cauchy in $L^1$, and so $S_n \to \sum_{j \geq 1}X_j$ in $L^1$ by completeness. Convergence in $L^1$ implies convergence in probability, but an increasing sequence converging in probability converges almost surely.  \end{proof}
The a.c.i.d.~condition in~\eqref{eq:defacid} coupled with summability of $\xi_n$ is equivalent to 
\begin{equation}
\sum_{n \geq 1} |\E\, \{\alpha_{n+1}(A) \given \calG_n\} - \alpha_n(A)     | < \infty\quad {\as},
    \notag
\end{equation}
for each $A \in \calB$ (see the discussion in Sect.~3 of \citet{battiston2025bayesian}). From the lemma it is then seen that the quasi-martingale condition of~\eqref{eq:defquasi} implies the a.c.i.d.~condition with $(\xi_n)_{n \geq 1}$ {\as} summable. That the quasi-martingale condition is in fact stronger than the a.c.i.d.~condition can be seen from the following simple example: For $n \geq 1$, let 
\begin{equation*}
    X_n = \begin{cases}n & \text{with probability }1/n^2 \\ 0 & \text{with probability }1 - 1/n^2.\end{cases}
\end{equation*}
For some $c > 0$, we see that $\Pr(X_n \geq c) = 1/n^2$ when $n \geq c$, and zero otherwise, so $\sum_{n\geq 1}\Pr(X_n \geq c) < \infty$. Hence $\sum_{n\geq 0} X_n < \infty$ {\as}~by the Borel--Cantelli lemma. On the other hand, $\E\,X_n = 1/n$, so $\sum_{n\geq 1}\E\,X_n = \infty$.  

\section{On the posterior mean}\label{app:posteriormean}
By absolute continuity, we know that the limiting measure admits a density and we can try to reason about the posterior mean of such an object. In particular, fixing an argument $y$, interest is on $\E\,[f(y)\given\calG_n]$. To start, we consider the posterior mean after only $N$ resampling steps, i.e. $\E\,[f_{n+N}(y)\given\calG_n]$, where $f_{n+N}(y)=(n+N)^{-1}\sum_{i=1}^{n+N} K_{n+N}(Y_i-y)$, where $K_{n+N}(\cdot)=h_{n+N}^{-1}K(\cdot / h_{n+N})$. Consider further the empirical measure $\PP_n = n^{-1}\sum_{i=1}^n \delta_{Y_i}$ and note that the KDE can be written as a convolution: $f_n(y) = (K_n * \PP_n)(y)$. The empirical measure further has a nice one-step-ahead updating equation
\begin{align*}
    \PP_{n+1} = \frac{n}{n+1}\PP_n + \frac{1}{n+1}\delta_{Y_{n+1}}.
\end{align*}
Now, defining the convolution operator $(H_n\nu)(dy) = (K_n * \nu)(y)dy$ as an operator on measures, noting in particular that $f_n(y)=(H_n\PP_n)(y)$, we use that $Y_{n+1} \given \calG_n \sim f_n$ and write
\begin{align*}
    \E\,[\PP_{n+1} \given \calG_n] &= \frac{n}{n+1}\PP_n + \frac{1}{n+1}\E\,[\delta_{Y_{n+1}} \given \calG_n] \\
    &=\frac{n}{n+1}\PP_n + \frac{1}{n+1}H_n\PP_n 
    =\underbrace{\left(\frac{n}{n+1}I + \frac{1}{n+1}H_n\right)}_{=:T_n}\PP_n.
\end{align*}
Here $I$ denotes the identity operator and we've defined the linear operator $T_n$ which can be regarded more generally as $T_m := (m/(m+1)I + 1/(m+1)H_m)$ for $m \geq n$. Now, taking one more predictive resampling step we consider $\E\,[\PP_{n+2} \given \calG_n] = \E\,\left[ \E\,[\PP_{n+2} \given \calG_{n+1}] \given\calG_n\right] 
    =\E\,\left[ T_{n+1}\PP_{n+1} \given\calG_n\right] =T_{n+1}T_n\PP_n$, from which repeated applications of the tower property provide the general formula 
    \begin{equation}
    \E\,[ \PP_{n+N} \given \calG_{n}]=T_{n+N-1}T_{n+N-2}\cdots T_n\PP_n.
        \notag
    \end{equation}
From here using that $f_{n+N}(y)=(H_{n+N}\PP_{n+N})(y)$, we have 
\begin{align*}
    \E\,[f_{n+N}(y) \given \calG_n] = (H_{n+N}T_{n+N-1}\cdots T_n \PP_n)(y).
\end{align*}
This above characterisation reduces the analysis of the posterior mean to one of the asymptotic behaviour of infinite products of linear operators, which we leave for future work.

\bibliographystyle{apalike}
\bibliography{refs}

\end{document}